%% file: main.tex
\documentclass[runningheads,orivec]{llncs}
\usepackage[ruled]{algorithm}
\usepackage[noend]{algpseudocode}
\usepackage{subcaption}
\usepackage{graphicx} 
\usepackage{lineno}
\usepackage[T1]{fontenc}
\usepackage{amssymb} % for \square
\usepackage{etoolbox}
\usepackage{booktabs} % for \toprule, \midrule, \bottomrule
\usepackage{makecell}

\makeatletter
\newcommand{\lncsqed}{%
  \ifhmode
    \qed
  \else
    \hfill\qed\par
  \fi
}
\makeatother

\let\lncsproof\proof
\let\endlncsproof\endproof
\renewenvironment{proof}{\lncsproof}{\lncsqed\endlncsproof}
\usepackage{tikz-cd}
\usepackage{array}
\usepackage{mdframed}
\usepackage{thm-restate}
\usepackage{pgf-pie}
 \usepackage{mathrsfs}
 \usepackage{enumerate}
 \usepackage{enumitem}

\usepackage{underscore}
\usepackage{pgfornament}
\usepackage{scalerel}
\usepackage{MnSymbol}
\usepackage{csquotes}
\usepackage{colortbl}
\usepackage{hyperref}
\usepackage{cleveref}
\usepackage{color}

\usepackage{afterpage} 
\usepackage{hhline}
\usepackage{multirow}
\usepackage{longtable}
\usepackage{adjustbox}
\usepackage[T1]{fontenc}
\usepackage[ttdefault=true]{AnonymousPro}

\input{mycommands}

\newcommand{\alttext}[1]{}

\begin{document}
%\linenumbers

\title{
Bringing closure to theory combination properties
}

\author{
Guilherme V. Toledo\inst{1} \and
Benjamin Przybocki\inst{2} \and
Yoni Zohar\inst{1}
}
\institute{
%\hfill
\begin{minipage}{0.4\textwidth}
\centering
\inst{1}Bar-Ilan University, Israel
\end{minipage}
\begin{minipage}{0.5\textwidth}
\centering
\inst{2}Carnegie Mellon University, USA
\end{minipage}
}
\maketitle              

\begin{abstract}
We consider the closure of three classical combination properties, namely, stable infiniteness, gentleness and shininess 
(or, equivalently for decidable theories, strong politeness), 
under intersection and combinability.
We compute every possible intersection, and then compute the maximal set of theories that can be combined with each resulting intersection.
We iterate this process until no new sets are identified.
How many properties will we end up with?

\end{abstract}

\section{Introduction}
\label{sec:intro}

\begin{table}[t]
\caption{Classical theory combination results and their sharpened versions}
\centering
\small
\setlength{\tabcolsep}{10pt}  % increase spacing
\begin{tabular}{@{}llll@{}}\toprule
\textbf{Source} & \textbf{Original Result} & \textbf{Sharpness Results~\cite{POPL}} \\ %& \textbf{More Results} \\
\midrule
%Nelson \& Oppen 
\cite{NelsonOppen}
& $\class{SI}\subseteq G(\class{SI})$
& $\class{SI}=G(\class{SI})$  \\
% & ---\\

%Tinelli \& Zarba 
\cite{shiny}
& $\class{shiny}\subseteq G(\class{})$
& $\class{shiny}=G(\class{})$,  $G(\class{shiny})=\class{}$\\

%Fontaine 
\cite{gentle}
& $\class{\G}\subseteq G(\class{\G})$
& $\class{CFS}=G(\class{\G})$,
 $G(\class{CFS})=\class{\G}$ \\
%,\quad \class{\G}=G(\class{CFS})$ \\
\bottomrule
\end{tabular}
\vspace{2mm}
\label{tab:sharp-summary}
\vspace{-5mm}
\end{table}

%\subsection{Background}
Theory combination in Satisfiability Modulo Theories 
\cite{BSST21,bonacina2019theory} studies how to combine decision procedures of two theories into a decision procedure for their combination (or, axiomatically, their union).
Most research in this field focuses on {\em disjoint} combination,
where the theories share no symbols other than equality.

% Denote by $G$, the operator that takes a
% set $X$ of decidable theories, and returns the set $G(X)$
% of decidable theories that can be disjointly combined with every theory of $X$.
% In other words,
% $G(X)$ is the set of theories $\T$ such that for every theory $\T'$ of $X$,
% the combination of $\T$ and $\T'$ (after shared symbols, if they exist, are renamed to be distinct) is decidable.
% %
% Three classical results in theory combination
% are given in the second column of \Cref{tab:sharp-summary}, where $\class{}$ denotes the set of all decidable theories,
% and 
% $\class{SI}$, $\class{shiny}$ and $\class{\G}$
% denote the sets of decidable theories that are, in addition,
% are stably infinite, shiny, and gentle, respectively. 

Nelson and Oppen proved \cite{NelsonOppen,OppenSI}
that any pair of decidable stably infinite\footnote{All technical notions in the introduction will be formally defined in \Cref{sec:prem}.} theories can be disjointly combined.
Tinelli and Zarba~\cite{shiny} proved that every decidable shiny theory can be disjointly combined with any decidable theory.\footnote{Following \cite{CasalRasga,polite,JB10-LPAR}, the same can be said about decidable strongly polite theories.}
Fontaine~\cite{gentle} proved that every decidable gentle theory can be disjointly combined with a general class of theories that includes gentle theories.%\footnote{In fact, Fontaine proved a stronger result, that allowed gentle theories to be combined with a larger set of theories.}

Recently, we studied the other directions of these results~\cite{POPL}.
We have shown that the first two results are {\em sharp}: every decidable theory that can be disjointly combined with all decidable stably infinite theories must be stably infinite;
and every decidable theory that can be disjointly combined with all decidable theories must be shiny.
As for gentleness, we proved in \cite{10.1007/978-3-031-99984-0_2,POPL} that the result of \cite{gentle} can be made sharp, by introducing the notion of theories with {\em computable finite spectra}.
% and proving that:
% We
% showed that this is sharp, in the sense that gentle theories are the largest set of decidable theories that can be combined with all decidable theories with computable finite spectra (and vice versa).
%

All these results are  summarized in \Cref{tab:sharp-summary} and \Cref{fig:not-hasse-diagram},
using the notations of \Cref{tab:notation-prop-th-def}.
% For example,
% $\class{}$ is the set of decidable theories, and
% $\class{SI}$ is the set of stably infinite theories.
They also employ a function, denoted $G$,
that takes a set $X$
of decidable theories, and returns the set $G(X)$
of all decidable theories that can be disjointly combined with every theory of $X$.
For example, the Nelson--Oppen theorem is phrased as
$\class{SI}\subseteq G(\class{SI})$ and its sharpened version as
$\class{SI}=G(\class{SI})$.

In \Cref{fig:not-hasse-diagram},
an edge between two sets of theories means that the lower is a subset of the upper (all inclusions were proven in \cite{POPL}).
% The dashed line 
% % marks a symmetry that 
% precisely captures $G$: 
$G(X)$ is  the  mirror of $X$ against the dashed line.
This structure 
 is contained in the lattice of sharp theory combination properties,
introduced in \cite{POPL}:
the order is $\subseteq$, the meet is the intersection, and the join (which is not as useful in our context) is defined in terms of the meet and $G$.
But, \Cref{fig:not-hasse-diagram} itself is not a lattice. 
For example, the intersection of gentleness and stable infiniteness is not there. 
The main motivating question of this paper is: 
{\bf what is the smallest lattice, closed under $G$, containing \Cref{fig:not-hasse-diagram}?}
The answer, as we prove in \Cref{sec:what-you-get}, is \Cref{fig:hasse-diagram}.

As a practical outcome,
we introduce two new combination methods in \Cref{sec:newcombmethods},
that allow combinations absent from 
 \Cref{fig:not-hasse-diagram}.
We improve the understanding of theory combination properties at a general, abstract level;
in the future, when theories are considered for implementation, researchers will have a larger suite of tests they can perform to determine the theories' combinability properties. 
And, if a pair of theories cannot be combined, researchers will have a more precise understanding of how the expressiveness of one of the theories needs to be weakened to allow combination.
Following \cite{POPL},
we focus on one-sorted logic, leaving the generalization to many-sorted logic to future work. %Such future work would first need to generalize the definition of the lattice from \cite{POPL}.

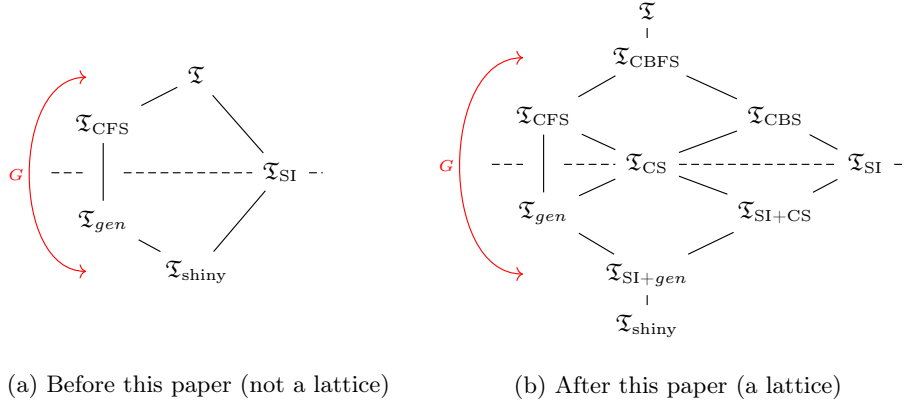
\begin{figure}[t]
  \centering

  \begin{subfigure}[t]{0.43\textwidth}
    \centering
    \alttext{A Hasse diagram of theory combination properties that were described before the present paper. The set of all decidable theories lies above both the set of theories with computable finite spectra and the set of stably infinite theories. The set of theories with computable finite spectra lies above the set of gentle theories. The set of gentle theories and the set of stably infinite theories lies above the set of shiny theories.}
    \begin{tikzcd}[sep=tiny]
      & \phantom{0} \arrow[bend right=90,swap,red,<->]{dddd}{\Gal} & \class{}\arrow[dash]{dl}{}\arrow[dash]{ddr}{}  & & & &\\
      & \class{CFS}\arrow[dash]{dd}{} & & &  &\\
      \phantom{O}\arrow[dashed,dash]{r}{}&\phantom{O}\arrow[dashed,dash]{rr}{}& & \class{SI}\arrow[dash]{ddl}{}\arrow[dash,dashed]{r} &\phantom{0}\\
      & \class{\G}\arrow[dash]{dr}{} & & &  & \\
      & \phantom{0} & \class{shiny} & &  & & 
    \end{tikzcd}
    \vspace{2.6em}
    \caption{Before this paper (not a lattice)}
    \label{fig:not-hasse-diagram}
  \end{subfigure}\hfill
  \begin{subfigure}[t]{0.53\textwidth}
    \centering
    \alttext{A Hasse diagram of theory combination properties, including those described in the present paper. The set of all decidable theories lies above the set of theories with computable bounded finite spectra, which lies above both the set of theories with computable finite spectra and the set of theories with computable bounded spectra. The set of theories with computable finite spectra lies above both the set of gentle theories and the set of theories with computable spectra. The set of theories with computable bounded spectra lies above both the set of theories with computable spectra and the set of stably infinite theories. The set of theories with computable spectra lies above both the set of gentle theories and the set of stably infinite theories with computable spectra. The set of stably infinite theories lies above the set of stably infinite theories with computable spectra. The set of gentle theories and the set of stably infinite theories with computable spectra lies above the set of stably infinite gentle theories, which lies above the set of shiny theories.}
    \begin{tikzcd}[sep=tiny]
      && \mathfrak{T}\arrow[dash]{d}{} &  & & &\\
      &\phantom{O}\arrow[bend right=90,swap,red,<->]{dddddd}{\Gal}& \class{CBFS}\arrow[dash]{ddl}{}\arrow[dash]{ddrr}{} & & & &\\
      & & & & & & \\
      &\class{CFS}\arrow[dash]{dr}{}\arrow[dash]{dd}{} & & & \class{CBS}\arrow[dash]{dr}{}\arrow[dash]{dll}{} &\\
      \phantom{O}\arrow[dashed,dash]{r}{}&\phantom{O}\arrow[dashed,dash]{r}{}& \class{CS}\arrow[dashed,dash]{rrr}{}\arrow[dash]{dl}{}\arrow[dash]{drr}{} & & & \class{SI}\arrow[dash]{dl}{}\arrow[dashed,dash]{r}{} & \phantom{O}\\
      &\class{\G}\arrow[dash]{ddr}{} & & & \class{SI+CS}\arrow[dash]{ddll}{} & \\
      & & & & & & \\
      & \phantom{O}& \class{SI+\G}\arrow[dash]{d}{} & &  & & \\
      & & \class{shiny} &  & & & \\
    \end{tikzcd}
    \caption{After this paper (a lattice)}
    \label{fig:hasse-diagram}
  \end{subfigure}

  \caption{Diagrams of theory combination properties
  % On the right is the smallest lattice that contains the structure on the left, which is itself not a lattice.
  }
  \vspace{-3mm}
\end{figure}

\paragraph*{\textbf{Outline.}}
In \Cref{sec:prem} we present preliminaries, from  first-order logic, through theory combination, to the Galois connection of theory combination.
\Cref{sec:what-you-get} presents the process through which we close the lattice;
\Cref{subsec:rel-between-prop} gives relationships between properties  and may be skipped in a first reading, being quite technical, much like
the new theories of \Cref{tab:extheories}, which are useful only for proofs. 
\Cref{sec:newcombmethods} gives further details on new combination methods.
\Cref{subsec:sharpness} shows sharpness.
\Cref{sec:odd} introduces a separating property  that we call
{\em self-combinability}, roughly meaning the ability of a theory to be combined with (a copy of) itself.
\Cref{sec:conclusion} concludes and sketches directions for future work.

\section{Preliminaries}
\label{sec:prem}

\subsection{First-order logic}
\label{sec:fol}
We review standard definitions regarding first-order logic \cite{Enderton}.
A \textbf{signature} $\Sigma$ is a set of function and predicate symbols,  each equipped with an arity, that includes $=$ as a binary predicate symbol.
We only consider signatures that are at most countably infinite.
Two signatures are \textbf{disjoint} if they share only the equality symbol.
% The \textbf{empty signature} has no functions and no predicates other than equality.
We define terms, formulas, literals, and sentences (formulas with no free variables) in the standard way;
$QF(\Sigma)$ is the set of quantifier-free $\Sigma$-formulas.

A \textbf{$\Sigma$-structure} $\mathbb{A}$ is a non-empty set $\dom{\mathbb{A}}$ (its domain) equipped with functions $f^{\mathbb{A}}:\dom{\mathbb{A}}^{n}\rightarrow\dom{\mathbb{A}}$ and predicates $P^{\mathbb{A}}\subseteq \dom{\mathbb{A}}^{m}$, for all 
function symbols $f$ and predicate symbols $P$ in $\Sigma$, of arities $n$ and $m$, respectively.
A \textbf{$\Sigma$-interpretation} $\A$ is an extension of a $\Sigma$-structure, where each variable $x$ is assigned a value $x^{\A}\in\dom{\A}$.
We define the value $\tau^{\A}\in\dom{\A}$ of a term, and the truth-value 
$\varphi^{\A}$
of a formula in an interpretation $\A$ in the usual way;
if $\varphi$ is true in $\A$ we write $\A\vDash\varphi$.
In \Cref{card-formulas} the formula $\neq(x_{1},\ldots,x_{n})$ for $n\geq 2$ states the variables $x_{1}$ through $x_{n}$ all take different values (for $n=1$ it is simply $true$).
The formulas $\psi_{\geq n}$, $\psi_{=n}$ and $\psi_{\leq n}$ state that there are  at least, exactly, and at most $n$ elements in the domain, respectively.

\begin{figure}[t]
\begin{mdframed}
\begin{equation*}
\begin{aligned}
    \neq(x_{1},\ldots,x_{n})= \bigwedge_{i=1}^{n-1}\bigwedge_{j=i+1}^{n}\neg(x_{i}=x_{j})\\
    \\\vspace{-4mm}
    \psi_{\geq n}= \Exists{{x_1,\ldots, x_n}}\NNEQ{x}
\end{aligned}
\quad
\begin{aligned}
\psi_{\leq n}= \Exists{x_1,\ldots,x_n}\Forall{y}\bigvee_{i=1}^{n}y=x_{i}\\
\\\vspace{-4mm}
\psi_{=n}=\psi_{\geq n}\wedge\psi_{\leq n}
\end{aligned}
\end{equation*}
\end{mdframed}
\vspace{-3mm}
\caption{Cardinality formulas}
\label{card-formulas}
\vspace{-5mm}
\end{figure}

A {{\textbf{$\Sigma$-theory}}} is a set of $\Sigma$-sentences (\textbf{axioms}).
A $\Sigma$-interpretation that satisfies all axioms of $\T$ is called a \textbf{$\T$-interpretation}, or a \textbf{model of $\T$}.
A formula is \textbf{$\T$-satisfiable} if there is a $\T$-interpretation that satisfies it, and a set of formulas is $\T$-satisfiable if there is a $\T$-interpretation that satisfies all formulas in the set simultaneously.
A $\Sigma$-theory is \textbf{decidable} when there is an algorithm that takes a quantifier-free $\Sigma$-formula and returns whether or not it is $\T$-satisfiable.
$\T_{Eq}$ is the theory over the empty signature (that only has $=$ as a symbol) having no axioms.
% commonly referred to as the theory of uninterpreted functions (UF).
Given signatures $\Sigma_{1}$ and $\Sigma_{2}$, $\Sigma_{1}\sqcup\Sigma_{2}$ is their {\bf disjoint union}, that is, if there are any symbols in common between the two, we rename them in $\Sigma_2$ before adding them to the symbols of $\Sigma_1$.
Given a $\Sigma_{1}$-theory $\T_{1}$ and a $\Sigma_{2}$-theory $\T_{2}$, their \textbf{disjoint combination} $\T_{1}\sqcup\T_{2}$ is the $\Sigma_{1}\sqcup\Sigma_{2}$-theory whose axiomatization is the union of $\T_{1}$ and $\T_{2}$ (after renaming the symbols in $\Sigma_{2}$ occurring in $\T_{2}$). 
% If $\T_{1}\sqcup\T_{2}$ is decidable we say $\T_{1}$ and $\T_{2}$ are \textbf{combinable}.

We denote $\mathbb{N}\cup\{\aleph_{0}\}$ by $\N$, and $\{n\in \N : p\leq n\leq q\}$ by $[p,q]$.
The \textbf{spectrum} of a formula $\varphi$ in a theory $\T$, denoted $\spec_{\T}(\varphi)$, is the subset of $\N$ of those elements $n$ such that there exists a $\T$-interpretation $\A$ that satisfies $\varphi$ with $|\dom{\A}|=n$.
Our restriction from the class of all cardinals to $\N$ in this definition is possible thanks to the downward L{\"o}wenheim--Skolem theorem, and the assumption that all signatures we deal with are countable. 
We state the 
L{\"o}wenheim--Skolem theorem as well as the compactness theorem:

\begin{theorem}[{\cite[Theorems~2.3.4 and 2.3.7]{marker2002}}] \label{LowenheimSkolem}
    Let $\Delta$ be a set of formulas over a countable signature, and let $\kappa \ge \aleph_0$. 
    Then, $\Delta$ is satisfied by an interpretation of size $\kappa$ if and only if it is satisfied by an interpretation of size $\aleph_0$.
\end{theorem}

%We will also make frequent use of the compactness theorem, found below.

\begin{theorem}[{\cite[Theorem~2.1.4]{marker2002}}] \label{compactness}
    A set of formulas is satisfiable if and only if every finite subset of it is satisfiable.
\end{theorem}

% In particular, we will utilize the following consequence of \Cref{LowenheimSkolem,compactness}.

\begin{restatable}{corollary}{corofcomp}
\label{compactness2}
Let $\Sigma$ be a countable signature, $\T$ a $\Sigma$-theory,
and $\phi$ a $\Sigma$-formula.
Suppose $\phi$ is $\T$-satisfiable
but $\aleph_0\notin\spec_{\T}(\phi)$.
Then $\spec_{\T}(\phi)$ is finite.
\end{restatable}

\subsection{Theory combination properties}
\label{sec:theorycombprops}

\begin{table}[t]
\caption{Notations for sets of theories}

\centering
\renewcommand{\arraystretch}{0.9}
\setlength{\tabcolsep}{6pt}

\begin{tabular}{@{}c p{5.5cm} c@{}}
\toprule
\textbf{Symbol} & \textbf{Description (decidable and \ldots)} & \textbf{Definition} \\
\midrule
$\class{}$ 
& --- 
& \Cref{sec:fol} \\

\midrule
$\class{SI}$ 
& Stably infinite 
& \multirow{7}{*}{\centering\Cref{sec:theorycombprops}} \\
$\class{shiny}$ 
& Shiny 
& \\
$\class{\G}$ 
& Gentle 
& \\
$\class{SI+\G}$ 
& Stably infinite and Gentle 
& \\
$\class{CFS}$ 
& Computable finite spectra 
& \\
$\class{CS}$ 
& Computable spectra 
& \\
$\class{SI+CS}$ 
& Stably infinite and computable spectra 
& \\

\midrule
$\class{CBFS}$ 
& Computable bounded finite spectra 
& \Cref{def:cbfs} \\
$\class{CBS}$ 
& Computable bounded spectra 
& \Cref{def:cbs} \\
\bottomrule
\end{tabular}
\vspace{2mm}
\label{tab:notation-prop-th-def}
\vspace{-8mm}
\end{table}

In \Cref{tab:notation-prop-th-def}, we summarize the theory combination properties of the paper:
essentially, we shall denote the set of decidable theories with property $X$ by $\class{X}$.
%(e.g., $\class{SI}$ is the set of stably infinite theories).

Let $\T$ be a theory, and let us denote $\mathbb{N}\setminus\{0\}$ by $\No$.
$\T$ is \textbf{stably infinite} \cite{NelsonOppen} if
$\spec_{\T}(\phi)\neq\emptyset$ implies $\aleph_0\in\spec_{\T}(\phi)$.
$\T$ has \textbf{computable finite spectra}~\cite{10.1007/978-3-031-99984-0_2}
if there is an algorithm that takes a quantifier-free formula $\phi$ and $n\in\No$, and returns whether $n\in\spec_{\T}(\phi)$;
$\T$ is \textbf{infinitely decidable} \cite{10.1007/978-3-031-99984-0_2} if there is an algorithm that takes a quantifier-free formula $\phi$ and returns whether $\aleph_{0}\in\spec_{\T}(\phi)$;
a theory with both these properties is said to have \textbf{computable spectra}.

$\T$ is \textbf{smooth} when, given a quantifier-free formula $\phi$, a $\T$-interpretation $\A$ that satisfies $\phi$, and any cardinal $\kappa\geq|\dom{\A}|$, there exists a $\T$-interpretation $\B$ that satisfies $\phi$ with $|\dom{\B}|=\kappa$.
$\T$ has the \textbf{finite model property} when, 
$\spec_{\T}(\phi)\neq\emptyset$ implies $\spec_{\T}(\phi)\cap\No\neq\emptyset$,
for every quantifier-free formula $\phi$.
And, if $\T$ is decidable, its \textbf{minimal model function} is the function $\minmod_{\T}$ that takes a quantifier-free formula $\phi$ and returns, if $\phi$ is $\T$-satisfiable, the least element in $\spec_{\T}(\phi)$.
$\T$ is \textbf{shiny} \cite{shiny} when it is smooth, has the finite model property, and its minimal model function is computable.

% \textbf{shiny} \cite{shiny} when there is an algorithm that takes a quantifier-free formula $\phi$ and returns an $n\in\mathbb{N}$ such that, if $\phi$ is $\T$-satisfiable, $\spec_{\T}(\phi)=[n,\aleph_{0}]$ (and, otherwise, $n$ may be chosen arbitrarily);\footnote{The usual definition requires the theory to be smooth, have the finite model property and a computable minimal model function \cite{LPAR}, but this definition is more succinct.}

Denote $\N\setminus\{0\}$ by $\N^{*}$.
$\T$ is \textbf{gentle} \cite{gentle,10.1007/978-3-031-99984-0_2} when there is an algorithm that takes a quantifier-free formula $\phi$ and returns a pair $(S,b)$, where $S\subset\No$ is finite
%\footnote{
%It is not enough that $S$ is simply computable and finite if we are to output it:
%if we were to simply test whether every $n\in\No$ is in $S$, there would be no way of knowing when we are done.
%If our algorithm is to return $S$, as a whole, it must in practice return an algorithm for testing whether $n\in S$, as well as an upper bound $m\in\mathbb{N}$ such that $n\in S$ implies $n\leq m$.} 
and $b$ is a Boolean, such that:
if $b$ is true, $\spec_{\T}(\phi)=S$;
and if $b$ is false, $\spec_{\T}(\phi)=\N^{*}\setminus S$.
The original definition of gentleness, however, is a bit different.
It says that, for every quantifier-free formula $\phi$:
$\spec_{\T}(\phi)$ is either a finite subset of $\No$ or a cofinite subset of $\N^{*}$ containing $\aleph_{0}$;
there is an algorithm that tells us which is the case (corresponding to our Boolean $b$);
and there is an algorithm that, given $\phi$, returns $\spec_{\T}(\phi)$ if this set is finite, and $\N^{*}\setminus\spec_{\T}(\phi)$ otherwise.
We prefer the former one, as it merges both algorithms into one.
% To make it easier to refer to properties of theories, we use symbols for each property, as listed in \Cref{tab:notation-prop-th-def}.
% All notations refer to decidable theories.
% For example, $\class{SI}$ denotes the set of all decidable stably infinite theories.
% %
 %From \cite{POPL}, we have:

\begin{proposition}[\cite{POPL}]\label{prop:known-relations-properties}
$\class{shiny}\subseteq\class{\G}\cap\class{SI}$ and
$\class{\G}\subseteq\class{CFS}$;
the theories of $\class{SI}$ are all
infinitely decidable; and
$\class{SI}\cap\class{CS}=\class{SI}\cap\class{CFS}$.
\end{proposition}

We call the following result \emph{Fontaine's lemma}. 
Its ``if'' direction was proved in \cite{gentle},
while its ``only if'' direction was proved in \cite{POPL}.
We could state it for arbitrary quantifier-free formulas, but we follow its original formulation.

\begin{lemma}[Fontaine's lemma]\label{lem-fontaine}
The disjoint combination $\T_1 \sqcup \T_2$ is decidable if and only if the following problem is decidable: given conjunctions of literals $\phi_1$ and $\phi_2$, determine whether $\spec_{\T_1}(\phi_1) \cap \spec_{\T_2}(\phi_2) = \emptyset$.
\end{lemma}

\subsection{The lattice of theory combination properties}
\label{sec:pre-gal}

\begin{table}[t]
\caption{Analogy between lattice-theoretic notions and theory combination}
\centering
\small
\setlength{\tabcolsep}{6pt}
\resizebox{\textwidth}{!}{%
\renewcommand{\arraystretch}{1.1}

\begin{tabular}{
  >{\raggedright\arraybackslash}p{2.1cm}
  >{\raggedright\arraybackslash}p{3.6cm}
  >{\raggedright\arraybackslash}p{5.7cm}
}
\toprule
\textbf{Notion} & \textbf{Lattice Theory} & \textbf{Theory Combination} \\
\midrule

$G$
& A Galois connection
& A function that maps $X$ to the set of all decidable theories combinable with $X$ \\
\midrule

$G$ is antitone
& $X \subseteq Y \Rightarrow G(Y) \subseteq G(X)$
& If $\T$ is combinable with $Y$, then it is combinable with every subset of $Y$ \\
\midrule

$G$ is a Galois connection
& $X \subseteq G(Y) \Leftrightarrow Y \subseteq G(X)$
& Combinability is symmetric \\
\midrule

$G \circ G$
& $X \subseteq G(G(X))$
& There is always a combination theorem between $X$ and $G(X)$ \\
\midrule

$X$ is closed
& $X = G(G(X))$
& The combination theorem between $X$ and $G(X)$ is sharp \\
\midrule

% $X \cap Y$ is closed
% & $X \cap Y = G(G(X \cap Y))$
% & There is a sharp combination theorem between $X \cap Y$ and $G(X \cap Y)$ \\
% \midrule

Closing a set under $\cap$ and $G$
& Finding the minimal containing lattice
& Finding the minimal set of properties with sharp combination theorems\\
\bottomrule
\end{tabular}
}
\vspace{2mm}
\label{tab:latcomb}
\end{table}

We review the main definitions of \cite{POPL}.
Two theories are {\bf combinable} 
if their disjoint combination is decidable.
Two sets of theories are combinable
if every theory of one is combinable with every theory of the other. 
We may identify an element with the singleton containing it.
Let $\class{}$ be the set of decidable theories.\footnote{This, in principle, would not constitute a set, but a proper class, but we can just select
all symbols for our (countable) signatures from a countable pool of  symbols.}
For each $X\subseteq\class{}$, $G(X)$ 
is the 
set of all decidable theories  combinable with $X$.

If $X\subseteq Y$ then every theory combinable with $Y$ is also combinable with $X$. 
Hence, $X\subseteq Y$ implies $G(Y)\subseteq G(X)$.
Further, combinability is symmetric: $X\subseteq G(Y)$ iff $Y\subseteq G(X)$.
Thus, $G$ is an \textbf{antitone Galois connection}. % w.r.t. $\subseteq$.
A {\bf combination theorem}  has the form 
$X\subseteq G(Y)$.
We always have a combination theorem between $X$ and $Y=G(X)$.
%($X\subseteq G(G(X))$ as $G$ is a Galois connection).
Such a theorem is {\bf sharp} if, moreover,
$X=G(Y)$.
In this case, $X$ is called {\bf closed under $G$}, or just {\bf closed}.
% For each $X$, $G(G(X))$ is the minimal
% closed superset of $X$.
% minimal -- if $X\subseteq Y\subseteq G(G(X))$
% with $Y=G(G(Y))$ then
% we prove $Y=G(G(X))$.
% Otherwise, $G(G(X))\not\subseteq G(G(Y))$,
% and since $G$ is a Galois connection,
% $G(Y)\not\subseteq G(X)$, and since $G$ is antitone,
% $X\not\subseteq Y$, which is a contradiction.
% Closed -- $X\subseteq G(G(X))$ and since $G$
% is antitone,
% $G(G(G(X)))\subseteq G(X)$ and again
% since $G$ is antitone,
% $G(G(X)) \subseteq G(G(G(G(X))))$.

Take any  set of closed sets of theories and for every
$X$ and $Y$ in it, add infimum $G(G(X\cap Y))$ and supremum $G(G(X \cup Y))$, until fixpoint.
The result is a {\bf lattice}
ordered by $\subseteq$.
If $X$ and $Y$ are closed, then the infimum and supremum are $X\cap Y$ and 
$G(G(X)\cap G(Y))$, resp.
% We always have $Z\subseteq G(G(Z))$ (in particular for $Z=X\cap Y$), and 
% since $X\cap Y\subseteq X,Y$ and $G$ is antitone,
% $G(X),G(Y)\subseteq G(X\cap Y)$ and so
% $G(G(X\cap Y))\subseteq G(G(X)),G(G(Y))$ and so
% $G(G(X\cap Y)\subseteq X\cap Y$.
%
% Further, if $X$ and $Y$ are closed,
% then $G(X\cup Y)=G(X)\cap G(Y)$.
% indeed, $X,Y\subseteq X\cup Y$ and since $G$
% is antitone, $G(X\cup Y)\subseteq G(X)\cap G(Y)$;
% And, since $G(X)\cap G(Y)\subseteq G(X)$ and
% $G(X)\cap G(Y)\subseteq G(Y)$
% and $G$ is antitone,
% $X=G(G(X))\subseteq G(G(X)\cap G(Y))$
% and
% $Y=G(G(Y))\subseteq G(G(X)\cap G(Y))$.
% $G$ is a Galois connection, and so 
% $G(X)\cap G(Y)\subseteq G(X)$ and
% $G(X)\cap G(Y)\subseteq G(Y)$.
Thus it suffices to compute closure 
under $\cap$ and $G$.
\Cref{tab:latcomb} summarizes the lattice-theoretic view of theory combination.

% To define what we call the Galois connection of theory combination, start with the set $\class{}$ of decidable theories\footnote{This, in principle, would not constitute a set, but a proper class, but we can just select
% all symbols for our (countable) signatures from a countable pool of  symbols.} and take its powerset $\mathscr{P}(\class{})$;
% this is a lattice under the order $\subseteq$, our starting lattice $(\mathscr{P}(\class{}),\subseteq)$.
% The Galois connection $G:\mathscr{P}(\class{})\rightarrow\mathscr{P}(\class{})$
% was defined in~\cite{POPL} as 
% \[G(X)=\{\T\in\class{} : \forall \T^{\prime}\in X.\T\sqcup\T^{\prime}\in \class{}\}.\]

% The lattice of its closed elements will be denoted by $(\mathcal{C},\subseteq)$, and of course it is a sublattice of our starting lattice $(\mathscr{P}(\class{}),\subseteq)$.
% %, what we denote by $(\mathcal{C},\subseteq)\leq(\mathscr{P}(\class{}),\subseteq)$.
% Notice that, although the order is the same in both, the infimum and supremum are not:
% while in $\mathscr{P}(\class{})$ these are simply the intersection and the union, in $\mathcal{C}$ we have $X\wedge Y=X\cap Y$ and $X\vee Y=G(G(X\cup Y))$.
% In this paper we will focus most of our efforts on the least lattice closed under $G$ and containing $\class{SI}$, $\class{\G}$ and $\class{shiny}$.
% %which we denote by $\langle \class{SI},\class{\G},\class{shiny}\rangle$.

\section{So, how many combination properties?}\label{sec:what-you-get}

\begin{algorithm}[t]
    \begin{algorithmic}[1]
        \Function{Generate}{$S$}
        \While{true}
            \State $S_{\cap} \gets \{X\cap Y\mid X,Y\in S\}$
            \State $S_{G} \gets \{G(X) \mid X \in S\}$
            \If{$ S = S \cup S_{\cap} \cup S_{G}$}
                 \Return $S$
            \EndIf
            \State $S\gets S \cup S_{\cap } \cup S_{G}$
        \EndWhile
        \EndFunction
    \end{algorithmic}
    \caption{Pseudocode for closure under $\cap$ and $G$.} \label{alg-G}
\end{algorithm}

\begin{table}[t]
\caption{Execution of \Cref{alg-G}}
\centering
\small
\renewcommand{\arraystretch}{0.9} % increase row spacing
\setlength{\tabcolsep}{6pt}       % horizontal padding

\begin{tabular}{ccccc}
\toprule
Iteration & Line & Computation & Result & Proof \\
\midrule
\multirow{3}{*}{0}
  & \multirow{3}{*}{---}
  & \multirow{3}{*}{---}
  & $\class{SI}$ & \multirow{3}{*}{---} \\
  & & & $\class{shiny}$ & \\
  & & & $\class{\G}$ & \\
\midrule
\multirow{3}{*}{1}
  & 3 & $\class{SI}\cap\class{\G}$ & $\class{SI+\G}$ & --- \\
  & 4 & $G(\class{shiny})$ & $\class{}$ & \Cref{tab:sharp-summary} \\
  & 4 & $G(\class{\G})$ & $\class{CFS}$ & \Cref{tab:sharp-summary} \\
\midrule
\multirow{2}{*}{2}
  & 3 & $\class{SI}\cap\class{CFS}$ & $\class{SI+CS}$ & \Cref{prop:known-relations-properties} \\
  & 4 & $G(\class{SI+\G})$ & $\class{CBFS}$ & \Cref{cor:gsigeqcbfs} \\
\midrule
3
  & 4 & $G(\class{SI+CS})$ & $\class{CBS}$ & \Cref{cor:gsicseqcbs} \\
\midrule
4
  & 3 & $\class{CFS}\cap\class{CBS}$ & $\class{CS}$ & \Cref{lem:SI->CBS} \\
\bottomrule
\end{tabular}
\vspace{2mm}
\label{tab:notation-prop-th-when}
\vspace{-5mm}
\end{table}

\Cref{alg-G}
takes a set $S$ of properties, and closes it
under $G$ and intersection.
As seen in \Cref{sec:pre-gal}, this  suffices
for computing the minimal lattice that contains $S$.
A priori, this process does not have to terminate.

Remarkably, it terminates when the input set 
is
$\{\class{SI}, \class{shiny}, \class{\G}\}$.
The resulting set is presented in \Cref{fig:hasse-diagram} as a Hasse diagram.
It mostly contains properties defined in \Cref{sec:prem}, though it also involves two new properties, namely {\em computable bounded finite spectra} and {\em computable bounded spectra}.
These give rise to two new combination theorems that we introduce in \Cref{sec:newcombmethods}.
%In other words:

\begin{theorem}
\label{thm:howmany}
\textsc{Generate}($\{\class{SI}, \class{shiny}, \class{\G}\}$) terminates, and its output consists of the 10 properties that appear in \Cref{fig:hasse-diagram}. 
\end{theorem}

\begin{proof}
The proof has two parts:
First, we show that each property in \Cref{fig:hasse-diagram} must indeed be in the resulting set, by justifying inclusion.
Second, we prove that there are no other added properties: after $4$ iterations of the loop of \Cref{alg-G}, we reach a fixed-point.

\medskip
\noindent
\underline{\bf Addition of properties:}
We justify the addition of each property in \Cref{fig:hasse-diagram} to the resulting set, by specifying in which step of \Cref{alg-G} it is added, and explaining why.
A tabular representation of this process is available in \Cref{tab:notation-prop-th-when}.

{\bf Iteration 0:}
Clearly, $\class{SI}$, $\class{shiny}$
and $\class{\G}$ must be included, as they are the input to the algorithm, and the algorithm never removes a property.

{\bf Iteration 1:}
The intersections of $\class{SI}$, $\class{shiny}$ and $\class{\G}$ are
added in line 3 of the first iteration of the loop.
Every shiny theory is stably infinite and gentle from \Cref{prop:known-relations-properties}, and thus we only have to add $\class{SI}\cap\class{\G}$,
denoted $\class{SI+\G}$.

Also, $G(X)$ must be added for each property $X$ in the input set, which is done in line 4.
By \Cref{tab:sharp-summary},
$G(\class{shiny})=\class{}$,
$G(\class{SI})=\class{SI}$,
and
$G(\class{\G})=\class{CFS}$.
Thus, $\class{}$ and $\class{CFS}$ are added (\class{SI} is already included).

{\bf Iteration 2:}
Now we close the $6$ properties collected so far under intersection.
Most of the resulting intersections are already
within these $6$ properties:
$\class{shiny}\subseteq\class{SI+\G}\subseteq\class{\G}\subseteq\class{CFS}\subseteq\class{}$,
and also $\class{shiny}\subseteq\class{SI+\G}\subseteq\class{SI}\subseteq\class{}$ (see \Cref{prop:known-relations-properties}), and 
the intersections of every two properties in the same
sequence of inclusions is just the smaller one.
Thus, we only
add a single property,
namely the intersection of $\class{SI}$ and $\class{CFS}$, that is,
$\class{SI}\cap\class{CFS}$. 
By \Cref{prop:known-relations-properties},
$\class{SI}\cap\class{CFS}=\class{SI}\cap\class{CS}$, and we choose to denote
this property 
by \class{SI+CS} because,
as will be seen below,
in the last iteration of the loop we will 
add $\class{CS}$, and the notation
$\class{SI+CS}$ will then better reflect
the placing of this property in \Cref{fig:hasse-diagram}. 

The next addition is due to the closure under $G$ computed in iteration 2.
We have already closed the three initial properties
under $G$ in iteration 1.
Also, $G(\class{})=\class{shiny}$
and $G(\class{CFS})=\class{\G}$, as shown 
in \Cref{tab:sharp-summary}.
Thus in this iteration, we only add
$G(\class{SI+\G})$.
This requires a new combination theorem, based
on a new property of theories,
that we call {\em computable bounded finite spectra},
denoted $\class{CBFS}$:
a theory has this property when it has an algorithm that, given a quantifier-free $\phi$ and $m,n\in\No$, returns whether $n$ is in $\spec_{\T}(\phi)$ as long as $\phi\wedge\neq(x_{1},\ldots,x_{m})$ (for fresh variables $x_{i}$) is not $\T$-satisfiable.  
It is formally defined in \Cref{def:cbfs},
and its sharp combination theorem is proven
in \Cref{cor:gsigeqcbfs}.

{\bf Iteration 3:}
The set obtained by the end of iteration 2 is already closed under intersection (see \Cref{prop:known-relations-properties,lem:CFS->CBFS,lem:SI->CBS}).
Thus, we close it under $G$.
From \Cref{tab:sharp-summary},
$G(\class{})=\class{shiny}$
and $G(\class{CFS})=\class{\G}$.
By \Cref{cor:gsicseqcbs},
$G(\class{CBFS})=\class{SI+\G}$.
Thus, we only
need to add $G(\class{SI+CS})$,
which requires yet another new combination theorem,
based on a new property, which we call
{\em computable bounded spectra},
denoted $\class{CBS}$:
fortunately, it can be defined as being infinitely decidable and having computable bounded finite spectra.
Its formal definition is in \Cref{def:cbs},
and a sharp combination theorem for it is proven
in \Cref{cor:gsicseqcbs}.

{\bf Iteration 4:}
In addition to the inclusions specified in Iteration 2,
we also have 
$\class{CFS}\subseteq\class{CBFS}\subseteq\class{}$ and
$\class{SI}\subseteq\class{CBS}\subseteq\class{CBFS}$
from \Cref{lem:CFS->CBFS,lem:SI->CBS},
and the intersection of every two properties in each resulting inclusion path is the smaller of the two.
The only missing intersection is that of 
$\class{CFS}$ and $\class{CBS}$,
which we prove in \Cref{lem:SI->CBS} to be equal to
$\class{CS}$, which we now add.

\medskip
\noindent
\underline{\bf There are no other properties:}
Denote by $S_0$ the set of 10 properties
listed in \Cref{fig:hasse-diagram}.
This set is $S_{\cap}$ after line 3 of the fourth
iteration of \Cref{alg-G}.
From \Cref{cor:gsigeqcbfs,tab:sharp-summary}, we have that after line 4 of the fourth iteration,
we get $S_G=S_0$ as well, that is,
that $G(S_0)=S_0$.
Going to the fifth iteration of the loop,
we have already seen that $S_0$ is  closed under intersection. Thus, in the fifth iteration,
after lines 3 and 4 are executed,
we have $S_{\cap}=S_G=S_0=S$,
and so in line 6 we return $S_0$.
\end{proof}

% \begin{figure}[h!]
%     \centering
%     \begin{tikzcd}[sep=small]
% && \mathfrak{T}\arrow[dash]{d}{} &  & & &\\
% &\phantom{O}\arrow[bend right=90,swap,red,<->]{dddddd}{\Gal}& \class{CBFS}\arrow[dash]{ddl}{}\arrow[dash]{ddrr}{} & & & &\\
% & & & & & & \\
% &\class{CFS}\arrow[dash]{dr}{}\arrow[dash]{dd}{} & & & \class{CBS}\arrow[dash]{dr}{}\arrow[dash]{dll}{} &\\
% \phantom{O}\arrow[dashed,dash]{r}{}&\phantom{O}\arrow[dashed,dash]{r}{}& \class{CS}\arrow[dashed,dash]{rrr}{}\arrow[dash]{dl}{}\arrow[dash]{drr}{} & & & \class{SI}\arrow[dash]{dl}{}\arrow[dashed,dash]{r}{} & \phantom{O}\\
% &\class{\G}\arrow[dash]{ddr}{} & & & \class{SI+CS}\arrow[dash]{ddll}{} & \\
% & & & & & & \\
% & \phantom{O}& \class{SI+\G}\arrow[dash]{d}{} & &  & & \\
% & & \class{shiny} &  & & & \\
% \end{tikzcd}
%     \caption{The Hasse diagram of theory combination properties starting from $\{\class{SI},\class{shiny},\class{\G}\}$ and closed under $G$ and intersections.}
%     \label{fig:hasse-diagram}
% \end{figure}

\subsection{Other resulting sets}

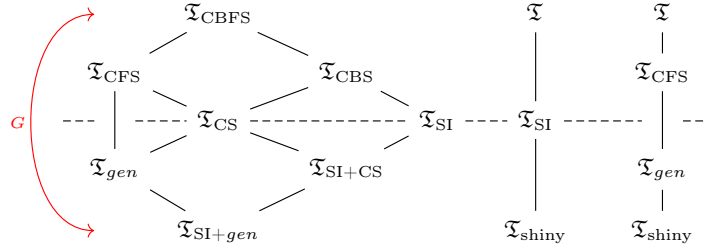
\begin{figure}[h!]
    \centering
    \alttext{Three Hasse diagrams of theory combination properties. The left diagram has all of the theory combination properties considered in this paper except for the set of all decidable theories and the set of shiny theories. The middle diagram has the set of all decidable theories above the set of stably infinite theories above the set of shiny theories. The right diagram has the set of all decidable theories above the set of theories with computable finite spectra above the set of gentle theories above the set of shiny theories.}
    \begin{tikzcd}[sep=tiny]
&\phantom{O}\arrow[bend right=90,swap,red,<->]{dddddd}{\Gal}& \class{CBFS}\arrow[dash]{ddl}{}\arrow[dash]{ddrr}{} & & & & & \class{}\arrow[dash]{ddd}{} & & & \class{}\arrow[dash]{dd}{} & \\
& & & & & & & & & & & \\
&\class{CFS}\arrow[dash]{dr}{}\arrow[dash]{dd}{} & & & \class{CBS}\arrow[dash]{dr}{}\arrow[dash]{dll}{} & & & & & & \class{CFS}\arrow[dash]{dd}{} & \\
\phantom{O}\arrow[dashed,dash]{r}{}&\phantom{O}\arrow[dashed,dash]{r}{}& \class{CS}\arrow[dashed,dash]{rrr}{}\arrow[dash]{dl}{}\arrow[dash]{drr}{} & & & \class{SI}\arrow[dash]{dl}{}\arrow[dashed,dash]{rr}{}& & \class{SI}\arrow[dashed,dash]{rrr}{}\arrow[dash]{ddd}{} & & & \phantom{O}\arrow[dashed,dash]{r}& \phantom{O}\\
&\class{\G}\arrow[dash]{ddr}{} & & & \class{SI+CS}\arrow[dash]{ddll}{} & & & & & &  \class{\G}\arrow[dash]{dd}{} & \\
& & & & & & & & & & & & & \\
& \phantom{O}& \class{SI+\G} & & & & & \class{shiny} & & & \class{shiny} & \\
\end{tikzcd}
    \caption{The Hasse diagrams for the resulting sets given different inputs}
    \label{fig:subdiagrams}
\end{figure}

We have generated a lattice from the three properties of stable infiniteness, gentleness and shininess.
What would happen if we were to restrict ourselves to only two among these properties instead? 

Stable infiniteness and gentleness are enough to generate most of \Cref{fig:hasse-diagram}:
a reasoning similar to the one found at 
the proof of \Cref{thm:howmany}
shows we only miss $\class{}$ and $\class{shiny}$.
For shininess and stable infiniteness, we only get $\class{}$, $\class{SI}$ and $\class{shiny}$.
For gentleness and shininess the situation is similar, except that we must add $G(\class{\G})=\class{CFS}$.
The resulting lattices are shown in \Cref{fig:subdiagrams}.
Formally:

\begin{theorem}~
\begin{enumerate}
\item \textsc{Generate}($\{\class{SI}, \class{\G}\}$) terminates, and its output consists of the 8 properties that appear in the left diagram of \Cref{fig:subdiagrams}.
\item \textsc{Generate}($\{\class{SI}, \class{shiny}\}$) terminates, and its output consists of the 3 properties that appear in the middle diagram of \Cref{fig:subdiagrams}.
\item \textsc{Generate}($\{\class{\G}, \class{shiny}\}$) terminates, and its output consists of the 4 properties that appear in the right diagram of \Cref{fig:subdiagrams}.
\end{enumerate}
\end{theorem}

% \subsection{Summary of the lattice}

% \begin{restatable}{theorem}{diagramcorrect}\label{theo:diagram-is-correct}\label{theo:structure}
%     The Hasse diagram in \Cref{fig:hasse-diagram} is correct;
%     that is, whenever $\class{X}$ is above $\class{Y}$, and there is an upwards path connecting the two, then $\class{X}\subseteq\class{Y}$;
%     and $\class{X}\cap\class{Y}$ is the lowest $\class{Z}$ sitting at least as high as $\class{X}$ and $\class{Y}$, with upward paths connecting $\class{X}$ to $\class{Z}$, and $\class{Y}$ to $\class{Z}$.
% \end{restatable}

\subsection{The relationships between the properties}\label{subsec:rel-between-prop}

\begin{table}[t]
\caption{Signatures}
\centering
\small
\setlength{\tabcolsep}{10pt}
\renewcommand{\arraystretch}{0.9}

\begin{tabular}{l c c}
\toprule
\textbf{Sig.} & \textbf{Functions} & \textbf{Predicates} \\
\midrule
$\Sigma_{1}$   & $\emptyset$ & $\emptyset$ \\[2pt]
$\Sigma_{P}$   & $\emptyset$ & $\{P\}$ \\[2pt]
$\Spn$         & $\emptyset$ & $\{P_{n} : n \in \No\}$ \\[2pt]
$\Spnn$        & $\emptyset$ & $\{P_{m,n} : m,n \in \No\}$ \\
\bottomrule
\end{tabular}
\vspace{2mm}
\label{tab:signatures-small}
\vspace{-5mm}
\end{table}

\begin{table}[t]
\caption{Theories defined over signatures from \Cref{tab:signatures-small}.
L.S.\ stands for the least set the theory belongs to.
$P_{\neq}=\{P_{i}\rightarrow\neg P_{j} : i\neq j\}$,
$P_{\infty}=\{P_{1}\rightarrow\psi_{\geq k} : k\in\mathbb{N}\}$, and
$P_{U}(n)=\{P_{n}\rightarrow \neg\psi_{=k} : k\notin U\}$.}
\centering
\small
\setlength{\tabcolsep}{4pt}
\renewcommand{\arraystretch}{1.15}

\begin{tabular}{
  l
  l
  p{7cm}
  c
  c
}
\toprule
\textbf{Sig.} & \textbf{Name} & \textbf{Axiomatization} & \textbf{Source} & \textbf{L.S.} \\
\midrule

$\Spn$ &
$\Th{si}$ &
$\{P_{n}\rightarrow \psi_{\geq m} : n\in \unc, m\in\No\}\cup P_{\neq}$ &
\cite{POPL} &
$\class{SI}$ \\[2pt]

$\Sp$ &
$\Th{cs}$ &
$\{P\rightarrow \psi_{=1}\}\cup\{\neg P\rightarrow \psi_{\geq m} : m\in\mathbb{N}\}$ &
\cite{POPL} &
$\class{CS}$ \\[2pt]

$\Sigma_{1}$ &
$\Tleqn$ &
$\{\psi_{\leq n}\}$ &
\cite{CADE} &
$\class{\G}$ \\

$\Sigma_{1}$ &
$\Tinfty$ &
$\{\psi_{\geq k} : k\in\No\}$ &
\cite{CADE} &
$\class{SI+CS}$ \\

$\Sigma_{1}$ &
$\Tgeqn$ &
$\{\psi_{\geq n}\}$ &
\cite{CADE} &
$\class{shiny}$ \\

\midrule

$\Sigma_{1}$ &
$\T_{\geq n+2}^{=n}$ &
$\{\psi_{=n}\vee\psi_{\geq n+2}\}$ &
New &
$\class{SI+\G}$ \\[2pt]

$\Spn$ &
$\T(m,n)$ &
$\{\psi_{=m}\vee\psi_{=n}\}
 \cup\{P_{k}\rightarrow\psi_{=m} : k\in U\}
 \cup\{P_{1}\rightarrow \psi_{=n}\}
 \cup P_{\neq}$ &
New &
$\class{}$ \\[2pt]

$\Spn$ &
$\Th{cbfs}$ &
$P_{\infty}
 \cup P_{U}(2)
 \cup\{P_{k}\rightarrow\psi_{\leq F(k)} : k>2, F(k)\in\mathbb{N}\}
 \cup P_{\neq}$ &
New &
$\class{CBFS}$ \\[2pt]

$\Spn$ &
$\Th{cfs}$ &
$P_{\infty}
 \cup\{P_{k}\rightarrow\psi_{\leq F(k)} : k>1, F(k)\in\mathbb{N}\}
 \cup P_{\neq}$ &
New &
$\class{CFS}$ \\[2pt]

$\Spn$ &
$\Th{cbs}$ &
$P_{U}(1)
 \cup\{P_{k}\rightarrow\psi_{=k} : k>1\}
 \cup P_{\neq}$ &
New &
$\class{CBS}$ \\

\bottomrule
\end{tabular}
\vspace{2mm}
\label{tab:extheories}
\vspace{-5mm}
\end{table}
To conclude this section, we prove that the properties
 are correctly situated in
  \Cref{fig:hasse-diagram}.
That is,
if there is a line between two properties,
the lower one is stronger than the upper one.
Most of these relationships are already captured by
\Cref{prop:known-relations-properties}.
Relationships that involve the new properties,
$\class{CBS}$ and $\class{CBFS}$, are proven in \Cref{lem:SI->CBS} below.
We just need two more connections:

\begin{restatable}{proposition}{newcontainments}
\label{thm:newcontainements}
$\class{SI}\cap\class{CFS}=\class{SI+CS}$ and $\class{\G}\cap\class{SI+CS}=\class{SI+\G}$.    
\end{restatable}

\begin{proof}
    By \Cref{prop:known-relations-properties}, stable infiniteness implies being infinitely decidable (for decidable theories), and by appealing to the definition of having computable spectra we get $\class{SI}\cap\class{CFS}=\class{SI+CS}$.

    By \Cref{prop:known-relations-properties}, gentleness implies having computable spectra, so the conjunction of gentleness, stable infiniteness and computable spectra is the same as that of gentleness and stable infiniteness, that is, $\class{\G}\cap\class{SI+CS}=\class{SI+\G}$.
\end{proof}

We also  prove that all inclusions in \Cref{fig:hasse-diagram} are strict. 
For this, we present theories that separate
between the different sets.
The signatures for these theories can be found in
\Cref{tab:signatures-small}.
$\Sigma_1$ is empty,
$\Sigma_{P}$ includes a single nullary predicate symbol $P$,
and $\Spn$ includes a nullary predicate symbol $P_n$ for every
$n\in\No$.
\Cref{tab:signatures-small} includes one more signature,
with predicate symbols for each pair of natural numbers,
to be used in \Cref{sec:newcombmethods} below.

The axiomatizations of the separating theories
are found in \Cref{tab:extheories}.
% For each $\class{X}$ there is a theory $\T$ for which the least set (among those found in \Cref{fig:hasse-diagram}, of course) it belongs to is precisely $\class{X}$.
In it,
$F:\No\rightarrow\No\cup\{\Inf\}$ is a function such that:
$(i)$ $\{(m,n) \in \No\times\No : F(m)\geq n\}$ is decidable; and $(ii)$ $\{n : F(n)=\aleph_{0}\}$ is undecidable.\footnote{Uses and explicit examples of such functions may be found in \cite{bonacina2019theory,10.1007/978-3-031-99984-0_2}.}
Also,
$\unc \subset \No$ is an undecidable set such that $1\notin U$.

The upper part of \Cref{tab:extheories} includes theories that were defined
in previous works (whose axiomatizations we present for completeness' sake).
Briefly, a $\Th{si}$-interpretation that satisfies a $P_{n}$, for $n\in U$, must be infinite;
a $\Th{cs}$-interpretation either satisfies $P$ and has a single element, or $\neg P$ and is infinite;
$\Tinfty$ has only infinite interpretations, $\T_{\leq n}$ only interpretations of size at most $n$, and $\T_{\geq n}$ interpretations of size at least $n$.

The last five theories will also be used in \Cref{sec:odd} below and are brand new.
$\T_{\geq n+2}^{=n}$ is easy to understand: it has models of cardinality exactly $n$, or strictly greater than $n+1$.
$\T(m,n)$ has models of cardinality $m$ or $n$;
if $P_{k}$ is true, for some $k\in U$, then the cardinality must be $m$;
but if $P_{1}$ is true (and we assumed $1\notin U$), then the cardinality must be $n$.
In $\Th{cbfs}$ the truth of:
$P_{1}$ implies a model is infinite;
$P_{2}$ implies a model has cardinality in $U$;
and $P_{k}$, for $k>2$, implies a model has cardinality at most $F(k)$, whenever $F(k)$ is finite.
$\Th{cfs}$ is similar to $\Th{cfbs}$, except that the truth of $P_{2}$ now implies a model has cardinality at most $F(2)$, if this value is finite.
And, finally, in $\Th{cbs}$ the truth of $P_{1}$ implies a model has cardinality in $U$, while the truth of $P_{k}$ for $k>1$ implies a model has cardinality precisely $k$.

Indeed,
the theories in
\Cref{tab:extheories} separate
the properties of \Cref{fig:hasse-diagram},
meaning that each has a property without having any of the properties below it.
This is represented in \Cref{tab:extheories} by
column L.S., which stands
for the ``least set'' of \Cref{fig:hasse-diagram} the theory belongs to.

\begin{restatable}{proposition}{exdiagramcorrect}\label{theo:ex-diagram-correct}
    The last column of \Cref{tab:extheories} correctly identifies the least set (L.S.) the theories belong to.
\end{restatable}

\begin{example}
Theory $\Th{si}$ is stably infinite but does
not have computable spectra.
Thus, it shows that the inclusion 
$\class{SI+CS}\subseteq\class{SI}$ observed
in \Cref{fig:hasse-diagram} is strict.
\end{example}

\section{Two new combination methods}
\label{sec:newcombmethods}
In this section, we introduce two new combination methods.
While they were tailored precisely for ``running'' \Cref{alg-G} in search of the smallest lattice that contains \Cref{fig:not-hasse-diagram},
they are interesting in their own right, and allow combinations of theories that other methods do not support (see \Cref{ex:newpossibility} below).

The new methods are inspired by gentleness and, in particular,
Fontaine's lemma (\Cref{lem-fontaine}) plays a crucial role in their correctness proofs.
% As seen in \Cref{sec:what-you-get}, the gentle combination method later led to the introduction of two properties of theories -- those that have computable finite spectra, and those that have computable spectra.
% Our new combination methods are based on two refinements of these notions:
% theories that have computable {\bf bounded} finite spectra, and those that have computable {\bf bounded} spectra.
% In both cases, dropping the ``finite'' from the name amounts to adding the requirement of infinite decidability.
% In addition to introducing these properties and using them for new combination methods, we study the relations between them and the other properties of \Cref{fig:hasse-diagram}.

\begin{algorithm}[t]
    \begin{algorithmic}[1]
        \Function{CBFS+(SI+GEN)}{$\phi_1, \phi_2$}
            \If{$\phi_2$ is not $\T_2$-satisfiable}
              \Return false
            \EndIf
            \State $m\gets \max(\N\setminus\spec_{\T_2}(\phi_2))$.
            \If{$\phi_1\wedge\neq(x_{1},\ldots,x_{m+1})$ is $\T_1$-satisfiable}\Comment{$x_{1},\ldots,x_{m+1}$ are fresh}
                \State \Return true
            \EndIf
            % \State $n \gets 1$
            \State{$X\gets \spec_{\T_1}(\phi_1)\setminus(\N\setminus\spec_{\T_2}(\phi_2))$}
            \For{$n\gets [1,m-1]$}
            %\While{$n\leq m-1$} 
                \If{$n\in X$}
                %\Comment{This set equals $\spec_{\T_1}(\phi_1)\cap\spec_{\T_2}(\phi_2)$}
                    \State \Return true
                \EndIf
%                \State $n \gets n+1$
            \EndFor
            \State \Return false
        \EndFunction
\end{algorithmic}
\caption{Determining non-emptiness of $\spec_{\T_1}(\phi_1)\cap\spec_{\T_2}(\phi_2)$, assuming
$\T_1\in\class{CBFS}$ and $\T_2\in\class{SI+\G}$.
}
 \label{alg-cbfssig}
\end{algorithm}

\begin{algorithm}[t]
    \begin{algorithmic}[1]
        \Function{(SI+CS)+CBS}{$\phi_1, \phi_2$}
            \If{$\phi_1$ is $\T_1$-satisfiable and $\aleph_0 \in \spec_{\T_2}(\phi_2)$}
                \State \Return true
            \EndIf
            \State $m \gets 1$
            \While{$\phi_2 \land \neq(x_{1},\ldots,x_{m})$ is $\T_2$-satisfiable} \Comment{$x_{1},\ldots,x_{m}$ are fresh}
                \State $m \gets m+1$
            \EndWhile
            \For{$n\gets [1,m-1]$}
                \If{$n\in\spec_{\T_1}(\phi_1)\cap\spec_{\T_2}(\phi_2)$}
                    \State \Return true
                \EndIf
             \EndFor   
            \State \Return false
        \EndFunction
\end{algorithmic}
\caption{Determining non-emptiness of $\spec_{\T_1}(\phi_1)\cap\spec_{\T_2}(\phi_2)$,
assuming $\T_1\in\class{SI+CS}$ and
$\T_2\in\class{CBS}$.
} \label{alg-sicscbs}
\end{algorithm}

\begin{figure}[t]
\centering
\alttext{An Euler diagram of five theory combination properties: computable bounded finite spectra (CBFS), computable finite spectra (CFS), computable bounded spectra (CBS), computable spectra (CS), and infinite decidability (ID). The inclusions represented by this diagram are stated as a proposition.}
\begin{tikzpicture}[scale=0.6]
\fill[blue!15] (2,1) rectangle (5,3);
\fill[red!10] (1,2) rectangle (4,4);
\fill[blue!55!red!20] (2,2) rectangle (4,3);
\draw[loosely dashed] (0,1) rectangle (5,5);
\draw[] (1,2) rectangle (4,4);
\draw[densely dotted, thick] (2,0) rectangle (6,3);
\node[] at (0.8,4.7) {CBFS};
\node[] at (5.6,0.3) {ID};
\node[] at (2.4,2.7) {CS};
\node[] at (4.4,1.3) {CBS};
\node[] at (1.6,3.7) {CFS};
\end{tikzpicture}
\caption{Relationships between having computable bounded finite spectra (CBFS, dashed box), having computable finite spectra (CFS, red), having computable bounded spectra (CBS, blue), having computable spectra (CS, purple), and being infinitely decidable (ID, dotted box)}
\label{fig:connections}
\vspace{-3mm}
\end{figure}
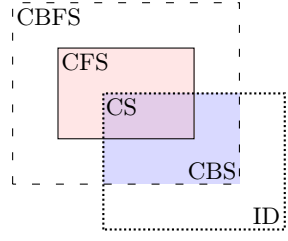

Recall that having computable finite spectra amounts to the existence of an algorithm that decides whether a given natural number is in the spectrum of a given formula.
We here weaken this requirement by allowing the algorithm to return arbitrary values in case the spectra is unbounded:
if it is bounded, however, it must answer correctly whether an element $n$ is or not in the spectra, given an adequate upper bound $m$.

\begin{definition}
\label{def:cbfs}
    A theory $\T$ is said to have {\em computable bounded finite spectra} if there is an algorithm that takes a quantifier-free formula $\phi$ and two $m,n\in\No$, and adheres to the following specification: if 
    $\phi$ is not $\T$-satisfiable by an interpretation
    with at least $m$ elements,
    then the algorithm must return 
    whether $n\in\spec_{\T}(\phi)$.
    Otherwise, the algorithm can return any Boolean value.
\end{definition}

Just like computable spectra adds infinite decidability to computable finite spectra,
a similar change to computable bounded finite spectra is pivotal:

\begin{definition}
\label{def:cbs}
    A theory $\T$ is said to have {\em computable bounded spectra} if it is infinitely decidable and has computable bounded finite spectra.
\end{definition}

We now examine the relations between these properties and others we have met so far.
The result is formalized in the next theorem, depicted in \Cref{fig:connections}.

\begin{restatable}{proposition}{SIimpliesCBS}
    \label{lem:SI->CBS}
    \label{lem:CFS->CBFS}    
    $\class{SI}\subseteq\class{CBS}\subseteq\class{CBFS}$, $\class{CFS}\subseteq\class{CBFS}$,
    and 
    $\class{CFS}\cap\class{CBS}=\class{CS}$.
\end{restatable}

\begin{proof}
To see that $\class{SI}\subseteq\class{CBS}$, given a quantifier-free $\phi$ and $m,n\in\No$, our algorithm may return an arbitrary boolean, as a stably infinite theory is infinitely decidable by \cite{POPL} and $\phi$ always has models of cardinality greater than $m$.

To see  that $\class{CFS}\subseteq \class{CBFS}$, given a quantifier-free formula $\phi$ and $m,n\in\No$, we can define the algorithm that proves $\T$ has computable bounded finite spectra by simply ignoring $m$ and feeding $\phi$ and $n$ to the algorithm that guarantees $\T$ has computable finite spectra.

That $\class{CBS}\subseteq\class{CBFS}$ follows from the definition of the latter.
Finally, as a theory with computable bounded spectra is infinitely decidable by definition, and computable spectra is defined as being infinitely decidable and having computable finite spectra, $\class{CBS}\cap\class{CFS}=\class{CS}$.
\end{proof}

The next theorem provides a new combination method between theories with computable bounded finite spectra, and theories both stably infinite and gentle.

\begin{restatable}{theorem}{IpartCBFS}
\label{thm:cbfs}
    If $\T_{1}\in\class{CBFS}$ and $\T_{2}\in\class{SI+\G}$, then $\T_{1}\sqcup\T_{2}$ is decidable.
    That is, $\class{CBFS}\subseteq G(\class{SI+\G})$ (and $\class{SI+\G}\subseteq G(\class{CBFS})$).
\end{restatable}

\begin{proof}[sketch]
    By \Cref{lem-fontaine}, it suffices to show a decision procedure
that decides, given $\phi_1$ and $\phi_2$ 
%that are
%$\Sigma_1$ and $\Sigma_2$ formulas, respectively,
whether $\spec_{\T_1}(\phi_1)\cap\spec_{\T_2}(\phi_2)$ is empty.
We present such a decision procedure in \Cref{alg-cbfssig}.
It relies on the properties of $\T_1$ and $\T_2$ in order to determine whether the models of the two formulas share a cardinality, by iterating over all possible cardinalities. Termination is guaranteed by \Cref{compactness2}.
Notice that in line 3, $m\in\No$:
since $\T_2$ is stably infinite we have $\aleph_{0}\in\spec_{\T_{2}}(\phi_{2})$, and since $\T_{2}$ is also gentle the set $\spec_{\T_2}(\phi_2)$ must be cofinite, making of $\N^{*}\setminus\spec_{\T_2}(\phi_2)$ a finite subset of $\No$.
\end{proof}

A similar yet distinct combination theorem is obtained
between theories that are stably infinite and have computable spectra, and theories with computable bounded spectra.
It is proved in a similar fashion to \Cref{thm:cbfs}.
In particular, 
it relies on \Cref{lem-fontaine} and \Cref{alg-sicscbs}.

\begin{restatable}{theorem}{IpartCBS}
\label{thm:cbs}
%Let $\T_1$ and $\T_2$ be $\Sigma_1$ and $\Sigma_2$ theories, respectively.
    If $\T_{1}\in\class{SI+CS}$  and $\T_{2}\in \class{CBS}$, then $\T_{1}\sqcup\T_{2}$ is decidable.
    That is, $\class{CBS}\subseteq G(\class{SI+CS})$ (and $\class{SI+CS}\subseteq G(\class{CBS})$).
\end{restatable}

Examples of theories with computable bounded finite spectra include all of those in \Cref{tab:extheories} but $\T(m,n)$, and examples of theories with computable bounded spectra include all of those in \Cref{tab:extheories} but $\T(m,n)$, $\Th{cbfs}$ and $\Th{cfs}$.
More importantly, $\Th{cbfs}$ is an example of a theory with computable bounded finite spectra but neither computable bounded spectra or computable finite spectra;
and $\Th{cbs}$ is an example of a theory with computable bounded spectra but that neither is stably infinite nor has computable spectra.
Such examples may seem artificial but finding a theory with a given property, and none of the stronger ones, can be non-trivial;
if we relax these restrictions, any stably infinite theory, for example, has both forms of bounded spectra.
The following example shows a combination made possible by \Cref{thm:cbs} not known to be possible by existing methods, in particular those in \Cref{fig:not-hasse-diagram}.

\begin{example}
\label{ex:newpossibility}
The combination of $\Th{cbs}$ and $\Tinfty$ 
from \Cref{tab:extheories}
is decidable
by \Cref{thm:cbs}, because the former has computable bounded spectra, and the latter is stably infinite and has computable spectra.
We don't get this from 
any of the combination methods that correspond
to \Cref{fig:not-hasse-diagram}:
not from Nelson--Oppen (as $\Th{cbs}$ is not stably infinite), nor from the shiny combination method (as neither of the theories is shiny), nor from the gentle combination method (as $\Th{cbs}$ does not have computable finite spectra).

\end{example}

\subsection{Sharpness}\label{subsec:sharpness}

% \begin{table}[t]
% \centering
% \renewcommand{\arraystretch}{1.15}
% \centering
% \begin{tabular}{|c|c|}\hline
% \phantom{O}Theory\phantom{O} & Property \\
% \hline
% \multirow{2}{*}{$\T_{>n}^{P}$} & Decidable, Computable Bounded  Spectra \\ 
% & $n=\max\spec_\T(\phi)\Rightarrow [\phi\wedge P_k \text{ is } \T\sqcup\T_{>n}^{P}\text{-SAT}\Leftrightarrow k\in U ] $
% \\\hline
% \multirow{2}{*}{$\T_{=}^{P}$} & 
% Decidable, Computable Bounded  Spectra \\
%  & $\phi\wedge P_n\text{ is } 
% \T\sqcup\T_{=}^{P}\text{-SAT}\Leftrightarrow 
% n\in\spec_{\T}(\phi)$\\\hline
% % $\T^{n}_{m}$ & ??? & ??? \\\hline
% \multirow{2}{*}{$\T_{\leq}^{S}$} & Decidable, Computable Bounded Finite Spectra \\ 
%  & $S=\No\setminus\spec_{\T}(\phi)\Rightarrow[\phi\wedge P_{k}\text{ is }\T\sqcup\T_{\leq}^{S}\text{-SAT}\Leftrightarrow F(k)=\aleph_{0}]$\\\hline
% \multirow{2}{*}{$\Tinfty$} & Decidable, SI+CS,  Computable Bounded Spectra \\
%  & $\phi\text{ is } 
% \T\sqcup\Tinfty\text{-SAT} \Leftrightarrow 
% \aleph_0\in\spec_{\T}(\phi)$\\\hline
% \multirow{2}{*}{$Th_\T$} & Decidable, Computable Bounded Finite Spectra\\
%  & $\phi\wedge P_{\phi,n}\text{ is }\T\sqcup Th_{\T}\text{-SAT}\Leftrightarrow n>|\No\setminus\spec_{\T}(\phi)|$  \\
% \hline\hline
% \multirow{2}{*}{$\T^{=}_{>n}$} & Decidable, Stably Infinite, Computable Spectra \\
%  & \phantom{OO}$n=\max\spec_\T(\phi)\Rightarrow [\phi\wedge P_k \text{ is } \T\sqcup\T_{>n}^{P}\text{-SAT}\Leftrightarrow k\in\spec_{\T}(\phi) ]$\phantom{OO}\\
% \hline
% \end{tabular}
% \caption{Properties of test theories. 
% }
% \label{tab:testtheoriesproperties}
% \end{table}

We now prove that \Cref{thm:cbfs,thm:cbs} are, in fact, sharp.
For \Cref{thm:cbfs}, this means
that every decidable theory that can be combined with all decidable theories that are both stably infinite and gentle, must have computable bounded finite spectra, and vice-versa.
A similar result holds for \Cref{thm:cbs}.

\begin{restatable}{theorem}{IIpartCBFS}
\label{thm:G-cbfs}
     $G(\class{SI+\G})\subseteq \class{CBFS}$
     and
     $G(\class{CBFS})\subseteq\class{SI+\G}$.
\end{restatable}

\begin{restatable}{theorem}{thmGcbs}
\label{thm:G-cbs}
    $G(\class{CBS})\subseteq\class{SI+CS}$ and $G(\class{SI+CS})\subseteq \class{CBS}$.
\end{restatable}

As a consequence, we obtain the following equalities, which are especially useful for the execution of \Cref{alg-G}:

\begin{restatable}{corollary}{GSIGeqCBFS}
\label{cor:gsigeqcbfs}
\label{cor:gsicseqcbs}
    $G(\class{SI+\G})=\class{CBFS}$,
    $G(\class{CBFS})=\class{SI+\G}$,
    $G(\class{SI+CS})=\class{CBS}$, and
    $G(\class{CBS})=\class{SI+CS}$.
\end{restatable}

The proofs of \Cref{thm:G-cbfs,thm:G-cbs}
are carried out using the technique
 of {\em test theories},
introduced in \cite{POPL}.
The idea is, given a reference set, to define a family of {\em test theories} that is, in a sense, ``maximally difficult'' to combine:
then, if a theory is combinable with these test theories, it must be combinable with the entire reference set.
We mostly use the test theories introduced in \cite{POPL}.
For them, we have to prove that they can indeed serve as test theories for the two new properties that we introduce.
We also introduce
one new theory, denoted $\Testcbfs$,
which is axiomatized in \Cref{axnewdouble} over signature
$\Spnn$
from \Cref{tab:signatures-small}.
Essentially, in a $\Testcbfs$-interpretation $\A$, if $P_{m,n}$ holds true then either $|\dom{\A}|=n$ or $|\dom{\A}|>m$.
It has the following properties:

\begin{figure}[t]
\begin{mdframed}
\centering
\vspace{1mm}
$\{P_{m,n}\rightarrow(\psi_{=n}\vee\psi_{\geq m+1}) : 1\leq n\leq m\}\cup \{P_{m,n}\rightarrow\neg P_{p,q} : (m,n)\neq (p,q)\}$
\vspace{1mm}
\end{mdframed}
\vspace{-5mm}
\caption{Axiomatization of $\Testcbfs$}
\label{axnewdouble}
\vspace{-3mm}
\end{figure}

\begin{restatable}{proposition}{propertiesoftestcbfs}\label{prop:prop-testcbfs}
    The theory $\Testcbfs$ is decidable, stably infinite, gentle, and for $\aleph_0\notin\spec_{\T}(\phi)$, $m=\max\spec_\T(\phi)$ finite, and $n \leq m$, we have that $\phi\wedge P_{m,n}$ is $\T\sqcup\Testcbfs$-satisfiable if and only if $n\in\spec_{\T}(\phi)$.
\end{restatable}

With this theory, and its  properties as listed in
\Cref{prop:prop-testcbfs}, we can 
 provide
a sketch for the proof of \Cref{thm:G-cbfs}.
The proof of \Cref{thm:G-cbs} is similar.

\begin{proof}[of \Cref{thm:G-cbfs}, sketch]
We show that $G(\class{SI+\G})\subseteq\class{CBFS}$.
The second part
can be done in a similar fashion, but requires further test theories from \cite{POPL}.

Suppose that $\T$ is decidable and can be combined with every decidable, gentle and stably infinite theory.
We prove that $\T$ has computable bounded finite spectra.
Take a quantifier-free $\phi$, and $m,n\in\No$:
        our algorithm simply returns whether $\phi\wedge P_{m,n}$ is $\T\sqcup\Testcbfs$-satisfiable;
        in case $\phi\wedge\neq(x_{1},\ldots,x_{m})$ is not $\T$-satisfiable this happens iff $n\in \spec_{\T}(\phi)$, for $1\leq n\leq m$, so the algorithm has the required properties.
        We are then done as $\T\sqcup\Testcbfs$ is indeed decidable, given that $\Testcbfs$ is decidable, gentle and stably infinite, see \Cref{prop:prop-testcbfs}.        
\end{proof}

\section{Theories that cannot be combined with themselves}\label{sec:odd}

In \Cref{fig:hasse-diagram},
$G$ behaves as a reflection over the dashed line.
% which makes visualizing which sets go together in a sharp combination theorem much easier:
% in particular, this means that 
For each set of theories $\class{X}$ in the diagram, either $\class{X}$ sits above $G(\class{X})$ (and so $G(\class{X})\subseteq\class{X}$), or $G(\class{X})$ sits above $\class{X}$ (and so $\class{X}\subseteq G(\class{X})$).
This is true for all combination properties we studied so far in this paper, and for all those found in \cite{POPL}, but is it true in general?
In other words, must $G$ look like a reflection?
The answer, surprisingly, is no, as we proceed to show.

We first introduce the notion of self-combinable theories, and prove that all theories up to the  dashed line of \Cref{fig:hasse-diagram}  are self-combinable.
This is an interesting property, and 
to our knowledge it was not studied so far.
We will use it below
to prove that $G$ does not have to be a reflection.

\begin{definition}
    A decidable theory $\T$ is {\em self-combinable} if
    $\T\sqcup\T$ is decidable.
\end{definition}

Recall that $\sqcup$ makes the signatures disjoint before joining the axiomatizations, and so this definition, while strange, is acceptable. 
We show below examples of both self-combinable and non self-combinable theories.

\begin{proposition}
\label{prop:self-comb-class-methods}
    If $\T\in\class{X}$
for $\class{X}\in\{\class{shiny},\class{SI+\G},\class{\G},\class{SI+CS},\class{CS},\class{SI}\}$ then it is self-combinable.
\end{proposition}

\begin{proof}
%We have
$\class{X}\subseteq G(\class{X})$,
and so $\T\in\class{X}\cap G(\class{X})$, hence $\T\sqcup\T$ is decidable.
\end{proof}

\begin{example}
\label{ex:self-comp}
We show the theory $\T(m,n)$ from \Cref{tab:extheories}, with $m>n$, is not self-combinable.
For simplicity, let $\Sigma_{P^{\prime}}^{\mathbb{N}}$ be a disjoint copy of $\Spn$, with predicates $P_{k}^{\prime}$ instead of $P_{k}$.
    Consider the quantifier-free formulas $P_{k}\wedge P_{1}^{\prime}$:
    they are $\T(m,n)\sqcup\T(m,n)$-satisfiable if and only if $k\notin U$.
    Indeed, if $k\in U$ then $P_{k}$ implies $\psi_{=m}$, while $P_{1}^{\prime}$ implies $\psi_{=n}$;
    and if $k\notin U$, define a $\T(m,n)\sqcup\T(m,n)$-interpretation $\A_{k}$ with $n$ elements, where $P_{k}$ is true, all other $P_{i}$ are false, $P_{1}^{\prime}$ is true, and all other $P_{i}^{\prime}$ are false;
    of course it satisfies $P_{k}\wedge P_{1}^{\prime}$.
    Thus, if $\T(m,n)\sqcup\T(m,n)$ is decidable, we get a decision procedure for the undecidable set $U$, a contradiction.
\end{example}

Now, $\T(m,n)$ is decidable, but has none of the other properties.
We show that for each property strictly above the dashed line of \Cref{fig:hasse-diagram} there
is a theory that exhibits it but is not self-combinable.
% Our understanding of $G$ and the existence of $\T(m,n)$ suggest that for any set of theories that lies above the dashed line in \Cref{fig:hasse-diagram} we should be able to find, in this set, a theory that cannot be combined with itself.
% This is indeed true.
We actually already defined these theories in the bottom part of \Cref{tab:extheories}:

\begin{restatable}{proposition}{bigselfcomb}
\label{prop:bigselfcomb} 
$\Th{cbfs}$
$\Th{cfs}$, and 
$\Th{cbs}$
 from
\Cref{tab:extheories} are not self-combinable.    
\end{restatable}

% here, however, we will only present the proof for $\T(m,n)$, which is simpler.

\medskip

Returning to the reflection principle, we obtain the following negative result:
\begin{proposition}
    There is a closed set $\class{X}$ with  $\class{X}\not\subseteq G(\class{X})$ and $G(\class{X})\not\subseteq \class{X}$.
\end{proposition}

\begin{proof}
Take $\class{X}$ to be $G(\{\T(m,n)\})$ for some $m>n$:
it is not in \Cref{fig:hasse-diagram} due to \Cref{prop:self-comb-class-methods}.
It is closed because every output of $G$ is closed, thus $G(G(\class{X}))=\class{X}$.
Notice that for any $p>q$ with $\{m,n\}\cap\{p,q\}=\emptyset$, $\T(m,n)$ is combinable with $\T(p,q)$, as the models of one do not share cardinalities with models of the other.

By \Cref{ex:self-comp}, $\T(m,n)$ is not in $\class{X}$ but $\T(p,q)$ is, and because of that $\T(m,n)$ is in $G(\class{X})$ but $\T(p,q)$ is not.
Therefore $\class{X}\not\subseteq G(\class{X})$
and $G(\class{X})\not\subseteq \class{X}$.
\end{proof}

\section{Conclusion}
\label{sec:conclusion}
We have computed the minimal lattice of theory combination properties that contains the classical properties of stable infiniteness, shininess, and gentleness. 
This process generated two new sharp combination methods, based on the novel notions of theories
with computable bounded (finite) spectra.

Going forward, we notice that \cite{POPL} showed the feasibility of proving sharpness of combination theorems.
We here adhere to this raised standard: the two new combination theorems that we present are proven sharp.
However, the sharpness proofs both here and in \cite{POPL} require tailored test theories.
We plan to streamline this process, looking for more general sufficient conditions for sharpness, and study the completeness of test theories, as defined in \cite{POPL}, for the new properties shown here.

\section*{Acknowledgments}
This research was supported in part by 
% the Stanford Center for Blockchain Research, 
Defense Advanced Research Projects Agency (DARPA) contract FA875024-2-1001,
Israel Science Foundation (ISF) grant number 209/25,
and the
Binational Science Foundation (BSF) grant number 2024049. Przybocki was supported by the NSF Graduate Research Fellowship Program under Grant No. DGE-2140739.

%\newpage

\bibliography{bib}

%\newpage

\appendix

\section{Proof of \Cref{compactness2}}
\corofcomp*

\begin{proof}
Since $\aleph_0\notin\spec_{\T}(\phi)$,
we have from \Cref{LowenheimSkolem} that
no infinite cardinality is in $\spec_{\T}(\phi)$.
Suppose for contradiction that $\spec_{\T}(\phi)$
is infinite.
Then every finite subset of 
the infinite set $S=\{\phi\wedge\phi_{\geq n}\mid n\in\spec_{\T}(\phi)\}$ is $\T$-satisfiable, and by
\Cref{compactness},
$S$ itself is $\T$-satisfiable. 
But no finite $\T$-interpretation can satisfy $S$.
\end{proof}

% \begin{proposition}\label{prop:new-test-theory}
%     $\T^{=}_{>n}$ is decidable, stably infinite and gentle.
% \end{proposition}

%\section{Proof of \Cref{lem:SI->CBS}}
%\SIimpliesCBS*

\section{Proof of \Cref{theo:ex-diagram-correct}}

\exdiagramcorrect*

\begin{proof}
We proceed from top to bottom of \Cref{fig:hasse-diagram}.

    \begin{enumerate}[leftmargin=1cm]
        \item $\T(m,n)$
        
        \begin{enumerate}
            \item We start by showing this theory is decidable.
            If $P_{1}$ occurs positively in $\phi$, we state $\phi$ is $\T(m,n)$-satisfiable if and only if $\minmod_{\T_{Eq}}(\phi^{\prime})\leq n$, where $\phi^{\prime}$ is the result of removing from $\phi$ any $P$-literals:
    one direction is easy, as a $\T(m,n)$-interpretation satisfying $\phi$ must have $n$ elements;
    for the other, we use $\T_{Eq}$ is smooth to get a $\T_{Eq}$-interpretation that satisfies $\phi^{\prime}$ with $n$ elements, and turn it into a $\T(m,n)$-interpretation by making $P_{1}$ true and all other $P$-literals false.

    If $\phi$ contains a positive $P$-literal other than $P_{1}$, or no positive $P$-literals at all, we state $\phi$ is $\T(m,n)$-satisfiable if and only if $\minmod_{\T_{Eq}}(\phi^{\prime})\leq m$.
    One direction is easy: 
    $\T(m,n)$ has interpretations of cardinality at most $m$.
    For the other direction, using that $\T_{Eq}$ is smooth (again) we can find a $\T_{Eq}$-interpretation that satisfies $\phi^{\prime}$ with cardinality $m$;
    this is made into a $\T(m,n)$-interpretation by making $P_{i}$ true if and only if it occurs positively in $\phi$, and it satisfies $\phi$ since $P_{1}$ does not occur positively in $\phi$.

    \item Although $P_{k}$ has no models of cardinality greater than $m$, we have that $n\in\spec_{\T(m,n)}(P_{k})$ if and only if $k\notin U$, so $\T(m,n)$ cannot have computable bounded finite spectra.
        \end{enumerate}

        \item $\Th{cbfs}$
        
        \begin{enumerate}
            \item First we prove decidability.
            If the only positive $P$-literal in $\phi$ is $P_{1}$ or $P_{2}$, or $\phi$ has no positive $P$-literals at all, then $\phi$ is $\Th{cbfs}$-satisfiable if and only if $\phi^{\prime}$ is $\T_{Eq}$-satisfiable.
    Indeed, we construct an infinite $\T_{Eq}$-interpretation that satisfies $\phi^{\prime}$ (notice that the consequent in both of the axioms $P_{1}\rightarrow\psi_{\geq k}$ and $P_{2}\rightarrow\neg\psi_{=k}$ are satisfied in an infinite interpretation), and then turn it into a $\Th{cbfs}$-interpretation by making $P_{i}$ true iff it occurs positively in $\phi$;
    of course the resulting interpretation satisfies $\phi$.

    Now suppose the only positive $P$-literal in $\phi$ is $P_{k}$, for some $k>2$.
    $\phi$ is then $\Th{cbfs}$-satisfiable iff $\minmod_{\T_{Eq}}(\phi^{\prime})\leq F(k)$ (remember the problem of whether $n\leq F(m)$ is decidable), and in that case its spectrum is $[\minmod_{\T_{Eq}}(\phi^{\prime}),F(k)]$.
    Indeed, given any $n$ in this interval, we use the smoothness of $\T_{Eq}$ to find a $\T_{Eq}$-interpretation that satisfies $\phi^{\prime}$ with $n$ elements, and transform it into a $\Th{cbfs}$-interpretation by making $P_{k}$ true in it, and all other $P_{i}$ false.
    The result satisfies $\phi$.

            \item  To summarize, we have just proved that if a conjunction of literals $\phi$ contains:
        $P_{1}$, $P_{2}$ or $P_{k}$, for a $k>2$ such that $F(k)=\aleph_{0}$, as its only positive $P$-literal, or no positive $P$-literals at all, then it has an infinite model, and in that case $\phi$ has models of cardinality greater than $m$ for any $m\in\No$;
        $P_{k}$, for a $k>2$ such that $F(k)\in\mathbb{N}$, as its only positive $P$-literal, then its spectrum is $[\minmod_{\T_{Eq}}(\phi^{\prime}),F(k)]$, which is a decidable set that does not contain $\aleph_{0}$.
        This implies $\Th{cbfs}$ has computable bounded finite spectra.

        \item $\Th{cbfs}$ does not have computable finite spectra as $k\in\spec_{\Th{cbfs}}(P_{2})$ iff $k\in U$.
        
        \item And it does not have computable bounded spectra as it is not infinitely decidable, since $\aleph_{0}\in \spec_{\Th{cbfs}}(P_{k})$, for $k>2$, if and only if $F(k)=\aleph_{0}$, leading to an undecidable problem.
        \end{enumerate}

        \item $\Th{cfs}$
        
        \begin{enumerate}
        \item Again we start with the decidability of the theory.
        If $P_{1}$ occurs positively in $\phi$, then $\phi$ is $\Th{cfs}$-satisfiable iff $\phi^{\prime}$ is $\T_{Eq}$-satisfiable:
    for the non-trivial direction, we once again use $\T_{Eq}$ is smooth to get an infinite $\T_{Eq}$-interpretation that satisfies $\phi^{\prime}$, make $P_{1}$ true in it and all other $P_{i}$ false.
    If $\phi$ contains no positive $P$-literals, we state its spectrum in $\Th{cfs}$ is $[\minmod_{\T_{Eq}}(\phi^{\prime}),\aleph_{0}]$ if $\phi^{\prime}$ is $\T_{Eq}$-satisfiable, and empty otherwise.
    Indeed, given a $k$ in this interval, we can find a $\T_{Eq}$-interpretation that satisfies $\phi^{\prime}$ with $k$ elements, and we then make all $P_{i}$ false.

    Now suppose the positive $P$-literal in $\phi$ is $P_{k}$, for $k>1$:
    we state $\phi$ is $\Th{cfs}$-satisfiable iff $\minmod_{\T_{Eq}}(\phi^{\prime})\leq F(k)$, what can be decided algorithmically, and in that case its spectrum is $[\minmod_{\T_{Eq}}(\phi^{\prime}),F(k)]$.
    Indeed, given any $n$ in this interval, we find a $\T_{Eq}$-interpretation that satisfies $\phi^{\prime}$ with $n$ elements, make $P_{k}$ true in it and all other $P_{i}$ false.

    \item In summary, the spectrum of a conjunction of literals $\phi$ in $\Th{cfs}$ equals:
        $\{\aleph_{0}\}$ if $\phi$ contains $P_{1}$;
        $[\minmod_{\T_{Eq}}(\phi^{\prime}),\aleph_{0}]$ if $\phi$ contains no positive $P$-literals;
        and, if $\phi$ contains a $P_{k}$ for $k>1$, $[\minmod_{\T_{Eq}}(\phi^{\prime}),F(k)]$.
        All of these sets are decidable, so $\Th{cfs}$ has computable finite spectra.

    \item To see that $\Th{cfs}$ does not have computable bounded spectra it is enough to notice $\Th{cfs}$ is not infinitely decidable:
    $P_{k}$, for $k>1$, has an infinite model iff $F(k)=\aleph_{0}$, which constitutes an undecidable problem.
             
        \end{enumerate}

        \item $\Th{cbs}$
        
        \begin{enumerate}
            \item First we prove $\Th{cbs}$ is decidable.
            If $\phi$ contains $P_{1}$ as positive $P$-literal, or no positive $P$-literals at all, then $\phi$ is $\Th{cbs}$-satisfiable if and only if $\phi^{\prime}$ is $\T_{Eq}$-satisfiable.
    For the non-trivial direction, we use that $\T_{Eq}$ is smooth to obtain an infinite $\T_{Eq}$-interpretation that satisfies $\phi^{\prime}$:
    we turn it into a $\Th{cbs}$-interpretation by making $P_{1}$ true iff it occurs positively in $\phi$, and all other $P_{i}$ false;
    of course the resulting interpretation satisfies $\phi$.

    If, instead, $\phi$ contains $P_{k}$ as a positive $P$-literal, for $k>1$, then we state $\phi$ is $\Th{cbs}$-satisfiable iff $\minmod_{\T_{Eq}}(\phi^{\prime})\leq k$.
    Indeed, using again $\T_{Eq}$ is smooth, we take a $\T_{Eq}$-interpretation that satisfies $\phi^{\prime}$ with $k$ elements, make $P_{k}$ true in it and all other $P_{i}$ false;
    the result is a $\Th{cbs}$-interpretation that satisfies $\phi$.

            \item To summarize, if a conjunction of literals $\phi$ contains $P_{1}$, or no positive $P$-literals at all, then it has an infinite model, and thus $\phi$ always has models of cardinality greater than $m$, for any $m\in\No$.
        If it has $P_{k}$, for a $k>1$, as its positive $P$-literal, then it does not have $\aleph_{0}$ in its spectrum, which equals $\{k\}$.
        Therefore, not only $\Th{cbs}$ is infinitely decidable, but it has computable bounded finite spectra.

        \item It is obvious $\Th{cbs}$ does not have computable finite spectra as $k\in\spec_{\Th{cbs}}(P_{1})$ iff $k\in U$.
        \end{enumerate}

        \item $\Th{cs}$ 
        
        This theory has computable spectra and is decidable, without being gentle nor stably infinite, as proven in \cite{POPL}.

        \item$\Th{si}$
        
        This theory was proven to be decidable and stably infinite, without having computable finite spectra, in \cite{POPL}.

        \item$\T_{\leq n}$
        
        \cite{POPL} proves this theory is decidable and gentle, without being stably infinite.

        \item$\Tinfty$ 
        
        This theory is proven in \cite{POPL} to have computable spectra and to be decidable and smooth (and thus stably infinite) but not gentle.
        It is not gentle as the spectra of all of its formulas is either empty or $\{\aleph_{0}\}$.

        \item $\T_{\geq n+2}^{=n}$
        
        \begin{enumerate}
            \item This theory is decidable as it satisfies the same quantifier-free formulas as $\T_{Eq}$:
            one direction of this is obvious;
            for the other we use the smoothness of $\T_{Eq}$ to obtain an interpretation with any size $m$ greater than or equal to $\minmod_{\T_{Eq}}(\phi)$ and satisfying the given formula $\phi$, which is easily seen to be a $\T_{\geq n+2}^{=n}$-interpretation whenever $m\in\{n\}\cup[n+2,\aleph_{0}]$, and there is always one such value.
            \item To see that $\T_{\geq n+2}^{=n}$ is stably infinite and gentle, we start by noticing 
            \[\spec_{\T_{\geq n+2}^{=n}}(\phi)=[\minmod_{\T_{Eq}}(\phi),\aleph_{0}]\cap(\{n\}\cup[n+2,\aleph_{0}]),\]
            as seen from the item above.
            This set is, of course, cofinite and contains $\aleph_{0}$, and we can compute a representation of its complement since $\minmod_{\T_{Eq}}(\phi)$ is computable.
            \item $\T_{\geq n+2}^{=n}$ is obviously not shiny, as it is not even smooth:
            it has models of size $n$ but none of size $n+1$.
        \end{enumerate}

        \item $\T_{\geq n}$
        
        \cite{POPL} proves this theory is decidable and shiny.
    \end{enumerate}
\end{proof}

\section{Proof of \Cref{thm:cbfs}}

\IpartCBFS*

\begin{proof}
As we mentioned before, by \Cref{lem-fontaine} it suffices to show a decision procedure
that decides, given $\phi_1$ and $\phi_2$, 
whether $\spec_{\T_1}(\phi_1)\cap\spec_{\T_2}(\phi_2)$ is empty.
This decision procedure was shown in \Cref{alg-cbfssig}.
To prove that it works, first, we show that each line in the algorithm can indeed be executed, and that the algorithm terminates.
Checking $\T_2$-satisfiability can be done as $\T_2$ is decidable.
Computing the maximal element of $\N\setminus\spec_{\T_2}(\phi_2)$ is possible because $\T_{2}$ is gentle and stably infinite, so $\spec_{\T_{2}}(\phi_2)$ is cofinite and there is an algorithm that outputs $\N\setminus \spec_{\T_2}(\phi_2)$.
It is possible to check for $\T_1$-satisfiability because $\T_1$ is decidable. 
And, checking whether $n\in\spec_{\T_1}(\phi_1)\setminus(\N\setminus\spec_{\T_2}(\phi_2))$
is possible because we have an algorithm that outputs $\N\setminus\spec_{\T_2}(\phi_2)$, coming from the fact $\T_2$ is gentle, and as we are in the case when $\phi\wedge\neq(x_{1},\ldots,x_{m+1})$ is not $\T_1$-satisfiable, $\spec_{\T_1}(\phi_1)\cap\mathbb{N}$ is computable.
The loop terminates as $m$ is finite, since $\spec_{\T_2}(\phi_2)$ is cofinite and contains $\aleph_{0}$.

Next, we prove correctness.

$(\Rightarrow)$:
First suppose 
$\spec_{\T_1}(\phi_1)\cap\spec_{\T_2}(\phi_2)$ is not empty.
Let $m=\max(\N\setminus\spec_{\T_{2}}(\phi_{2}))$: 
if $\spec_{\T_1}(\phi_1)\cap\spec_{\T_2}(\phi_2)\cap[m+1,\aleph_{0}]$ is not empty, then we will get true from line $5$;
otherwise $\spec_{\T_1}(\phi_1)\cap\spec_{\T_2}(\phi_2)\cap[1,m-1]$ is not empty, and our loop will find exactly the first element in this set.

$(\Leftarrow)$:
Now suppose 
$\spec_{\T_1}(\phi_1)\cap\spec_{\T_2}(\phi_2)$ is  empty.
Therefore we get that $\phi_1\wedge\neq(x_1,\ldots,x_{m+1})$ is not $\T_1$-satisfiable, for $m=\max\N\setminus\spec_{\T_{2}}(\phi_{2})$;
and for no $1\leq n\leq m-1$ we have $n\in\spec_{\T_1}(\phi_1)\setminus(\N\setminus\spec_{\T_2}(\phi_2)$, as this set equals $\spec_{\T_1}(\phi_1)\cap\spec_{\T_2}(\phi_2)$.
So the algorithm returns false.
\end{proof}

\section{Proof of \Cref{thm:cbs}}

\IpartCBS*

\begin{proof}
By \Cref{lem-fontaine}, it suffices to show a procedure that,
given $\phi_1$ and $\phi_2$, decides whether
$\spec_{\T_1}(\phi_1)\cap\spec_{\T_2}(\phi_2)$ is empty.
\Cref{alg-sicscbs} is an adequate algorithm.
First, we show that each line in the algorithm can indeed be executed, and that the algorithm terminates.
It is possible to check whether 
$\aleph_0 \in \spec_{\T_2}(\phi_2)$, as
$\T_2$ has computable bounded spectra, and therefore is infinitely decidable.
The $\T_2$-satisfiability check for
$\phi_2 \land \neq(x_{1},\ldots,x_{m})$
is also possible as $\T_2$ is decidable.
The loop must terminate, as we must have a maximal
$m$ in $\spec_{\T_2}(\phi_2)$ by \Cref{compactness2}.
Checking whether $n\in\spec_{\T_1}(\phi_1)$ is possible
because $\T_1$ has computable spectra.
Checking whether $n\in\spec_{\T_2}(\phi_2)$
is possible because we are in the case where $\phi_2 \land \neq(x_{1},\ldots,x_{m})$ is $\T_2$-unsatisfiable, and $\T_2$
has computable bounded spectra.

Next, we prove correctness.

$(\Rightarrow)$:
First suppose 
$\spec_{\T_1}(\phi_1)\cap\spec_{\T_2}(\phi_2)$ is not empty.
If $\aleph_0$ is in this intersection then it is in particular in $\spec_{\T_2}(\phi_2)$ and so the algorithm will return true. 
Otherwise, there is some $n\in\No$ in the intersection, which will be found during the loop.

$(\Leftarrow)$:
Now suppose 
$\spec_{\T_1}(\phi_1)\cap\spec_{\T_2}(\phi_2)$ is  empty.
In particular, $\aleph_0$ is not in the intersection.
$\T_1$ is stably infinite and so $\aleph_0\in\spec_{\T_1}(\phi_1)$, and therefore
$\aleph_0\notin\spec_{\T_2}(\phi_2)$.
Also, at no iteration of the loop an $n$ in the intersection will be found, and hence false
 will be returned.
\end{proof}

\section{Proof of \Cref{prop:prop-testcbfs}}

\propertiesoftestcbfs*

\begin{proof}
    Much like we have done several times before, a $P$-literal is now a literal of the form $P_{m,n}$ or $\neg P_{p,q}$, the former being also called positive.
    If a conjunction of literals $\phi$ has a $P$-literal and its negation, or more than one positive $P$-literal, it is clear that it is not $\Testcbfs$-satisfiable from the axioms $\{P_{m,n}\rightarrow\neg P_{m^{\prime},n^{\prime}} : (m,n)\neq(m^{\prime},n^{\prime})\}$.
    So assume $\phi$ contains at most one positive $P$-literal $P_{m,n}$, does not contain its negation, and $\phi^{\prime}$ is the result of removing from $\phi$ any $P$-literals.

    If $\phi$ contains $P_{m,n}$, for $n>m$, we get simply $\phi$ is $\Testcbfs$-satisfiable if and only if $\phi^{\prime}$ is $\T_{Eq}$-satisfiable, and its spectra is then $[\minmod_{\T_{Eq}}(\phi^{\prime}),\aleph_{0}]$:
    after all, there are no axioms governing the behavior of $P_{m,n}$ in this case.
    If $\phi$ contains $P_{m,n}$, for some $1\leq n\leq m$, the interesting case, we state its spectrum equals $[M,\aleph_{0}]$, if $n<\minmod_{\T_{Eq}}(\phi^{\prime})$, and $\{n\}\cup[M,\aleph_{0}]$ otherwise, where $M=\max\{m+1,\minmod_{\T_{Eq}}(\phi^{\prime})\}$.
    Indeed, for any cardinality in these sets, take a $\T_{Eq}$-interpretation that satisfies $\phi^{\prime}$, make $P_{m,n}$ true and all other $P_{p,q}$ false;
    the resulting interpretation is a $\Testcbfs$-interpretation that satisfies $\phi$.
    Obviously $\Testcbfs$ is then seen to be decidable, stably infinite ($\aleph_{0}$ is always in the spectra), and gentle (as the spectrum is cofinite and it is easy to output its complement).

    Finally, suppose $\aleph_{0}\notin\spec_{\T}(\phi)$ and that $m=\max\spec_{\T}(\phi)$.
    If $1\leq n\leq m$ is in $\spec_{\T}(\phi)$ there is a $\T$-interpretation $\A$ that satisfies $\phi$ with $n$ elements in its domain;
    by making $P_{m,n}$ true and all other $P_{p,q}$ false this becomes a $\Testcbfs$-interpretation (and thus a $\T\sqcup\Testcbfs$-interpretation) satisfying $\phi\wedge P_{m,n}$.
    Conversely, suppose $\phi\wedge P_{m,n}$ is satisfied by the $\T\sqcup \Testcbfs$-interpretation $\A$:
    because $\A$ satisfies $P_{m,n}$ and is a $\Testcbfs$-interpretation it has either $n$ or more than $m$ elements;
    because $\A$ is a $\T$-interpretation that satisfies $\phi$ it has at most $m$ elements;
    so it has precisely $n$ elements, and therefore $n\in\spec_{\T}(\phi)$.
\end{proof}

\section{Test theories for the proofs of \Cref{thm:G-cbfs,thm:G-cbs}}

\begin{table}[t]
\centering
\centering
\begin{tabular}{|c|c|c|}
\hline
Sig. & \phantom{O}Functions\phantom{O} & Predicates\\
\hline
$\Sigma_{1}$ & $\emptyset$ & $\emptyset$\\
$\Sigma_{P}$ & 
$\emptyset$ & $\{P\}$\\
$\Spn$ & $\emptyset$ & $\{P_{n} : n\in \No\}$\\
$\Spnn$ & $\emptyset$ & $\{P_{m,n} : m,n\in\No\}$\\
\phantom{O}$\Sigmasigma$\phantom{O} & $\emptyset$ & \phantom{O}$\{P_{\phi,n} : \phi\in QF(\Sigma), n\in\No\}$\phantom{O}\\\hline
\end{tabular}
\caption{Signatures}
\label{tab:signatures}
\end{table}

\begin{table}[!htbp]
\centering
\renewcommand{\arraystretch}{1.15}
\centering
\begin{tabular}{|c|c|>{\centering\arraybackslash}p{10cm}|}\hline
\phantom{O}Th.\phantom{O} & \phantom{O}Sig.\phantom{O} & Details 
\\\hline
$\T_{>n}^{P}$ & $\Spn$ &   
\makecell{
\underline{Axiomatization}\\
$\{P_{m}\rightarrow \psi_{\geq n+1} : m\in \unc\}\cup P_{\neq}$ \\[.5em]
\underline{Properties} \\
Decidable, Computable Bounded  Spectra \\ 
 $\aleph_{0}\notin \spec_{\T}(\phi), n=\max\spec_\T(\phi)\Rightarrow [\phi\wedge P_k \text{ is } \T\sqcup\T_{>n}^{P}\text{-SAT}\Leftrightarrow k\notin U ] $\\[.5em]
 \underline{Usage}\\
 \Cref{thm:G-cbfs,thm:G-cbs}
}\\\hline
$\T_{=}^{P}$ & $\Spn$ &  
\makecell{
\underline{Axiomatization}\\
$\{P_{n}\rightarrow\psi_{=n} : n\in\No\}$ \\[.5em]
\underline{Properties} \\
Decidable, Computable Bounded  Spectra \\
  $\phi\wedge P_n\text{ is } 
\T\sqcup\T_{=}^{P}\text{-SAT}\Leftrightarrow 
n\in\spec_{\T}(\phi)$\\[.5em]
 \underline{Usage}\\
 \Cref{thm:G-cbfs,thm:G-cbs}
}\\\hline
% $\T^{n}_{m}$ & $\Spn$ &  
% \makecell{
% \underline{Axiomatization}\\
% $\{\psi_{=m}\vee\psi_{=n}\}\cup\{P_{k}\rightarrow\psi_{=n} : k\in \unc\}\cup P_{\neq}$ \\[.5em]
% \underline{Properties} \\[.5em]
%  \underline{Usage}\\
% }\\\hline
$\T_{\leq}^{S}$ & $\Spn$ &  
\makecell{
\underline{Axiomatization}\\
$\{P_{n}\rightarrow\psi_{\leq F(n)} : F(n)\in\mathbb{N}\}\cup\{\neg\psi_{=n} : n\notin S\}\cup P_{\neq}$\\[.5em]
\underline{Properties} \\
Decidable, Computable Bounded Finite Spectra \\ 
  $\aleph_{0}\in \spec_{\T}(\phi), S=\No\setminus\spec_{\T}(\phi), |S|\geq \aleph_{0}\Rightarrow$\\$[\phi\wedge P_{k}\text{ is }\T\sqcup\T_{\leq}^{S}\text{-SAT}\Leftrightarrow F(k)=\aleph_{0}]$\\[.5em]
 \underline{Usage}\\
 \Cref{thm:G-cbfs}
}\\\hline
$\Tinfty$ & $\Sigma_{1}$ &  
\makecell{
\underline{Axiomatization}\\
$\{\psi_{\geq n} : n\in\No\}$\\[0.5em]
\underline{Properties} \\
Decidable, SI+CS,  Computable Bounded Spectra \\
  $\phi\text{ is } 
\T\sqcup\Tinfty\text{-SAT} \Leftrightarrow 
\aleph_0\in\spec_{\T}(\phi)$\\[.5em]
 \underline{Usage}\\
 \Cref{thm:G-cbs}
}\\\hline
$Th_\T$ & $\Sigmasigma$ &  
\makecell{
\underline{Axiomatization}\\
\Cref{def:theorytheory}\\[0.5em]
\underline{Properties} \\
Decidable, Computable Bounded Finite Spectra\\
$\phi\wedge P_{\phi,n}\text{ is }\T\sqcup Th_{\T}\text{-SAT}\Leftrightarrow n>|\No\setminus\spec_{\T}(\phi)|$\\[.5em]
%\\[.5em]
 \underline{Usage}\\
 \Cref{thm:G-cbfs}
}\\\hline
\hline
\makecell{
$\Testcbfs$ \\
(new)
}
& $\Spnn$ &
\makecell{%
\underline{Axiomatization}\\
$\{P_{m,n}\rightarrow(\psi_{=n}\vee\psi_{\geq m+1}) : 1\leq n\leq m\}\cup \{P_{m,n}\rightarrow\neg P_{p,q} : (m,n)\neq (p,q)\}$\\[0.5em]
\underline{Properties} \\
Decidable, Stably Infinite, Gentle \\
 $\aleph_{0}\notin \spec_{\T}(\phi), m=\max\spec_\T(\phi)\Rightarrow$\\$[\phi\wedge P_{m,n} \text{ is } \T\sqcup\Testcbfs\text{-SAT}\Leftrightarrow n\in\spec_{\T}(\phi) ]$\\[.5em]
 \underline{Usage}\\
 \Cref{thm:G-cbfs,thm:G-cbs}
}\\\hline
\end{tabular}
\caption{Test theories. 
$F:\No\rightarrow\No\cup\{\Inf\}$ is a function such that:
$(i)$ $\{(m,n) \in \No\times\No : F(m)\geq n\}$ is decidable; and $(ii)$ $\{n\in\No : F(n)=\aleph_{0}\}$ is undecidable.
$\unc \subset \No$ is an undecidable set.
Let $P_{\neq}$ denote the set $\{P_{i}\rightarrow\neg P_{j} : i\neq j\}$.
}
\label{tab:testtheories}
\label{tab:testtheoriesproperties}
\end{table}

The test theories are presented in \Cref{tab:testtheories}.
In it, $F:\No\rightarrow\No\cup\{\Inf\}$ is a function such that:
$(i)$ $\{(m,n) \in \No\times\No : F(m)\geq n\}$ is decidable; 
and $(ii)$ $\{n\in\No : F(n)=\aleph_{0}\}$ is undecidable.
One example of such a function takes the index of a Turing machine (according to some enumeration) and returns the number of steps it takes to halt (or $\aleph_0$ if it does not halt). Let $\unc \subset \No$ be an undecidable set.
We will use the signatures found in \Cref{tab:signatures}.
For $\Spn$ we will call literals of the form $P_{i}$ or $\neg P_{j}$ $P$-literals, the former being also said to be positive;
in $\Sigmasigma$ the $P$-literals are those of the form $P_{\phi,n}$ or $\neg P_{\psi,m}$.

For $\T^{P}_{>n}$, a theory dependent on an $n\in\No$, if $P_{m}$ is true in an interpretation, for some $m\in U$, then the interpretation has at least $n+1$ elements;
that is, $\spec_{\T^{P}_{>n}}(P_{m})=[n+1,\aleph_{0}]$ if $m\in U$, and $[1,\aleph_{0}]$ otherwise.
Now a $\T^{P}_{=}$-interpretation where $P_{n}$ is true must have exactly $n$ elements;
so $\spec_{\T^{P}_{=}}(P_{n})=\{n\}$.
% The theory $\T_{m}^{n}$ depends on two numbers $m$ and $n$, and we will usually assume that $m<n$: 
% an interpretation for this theory must have either $m$ or $n$ elements, and if $P_{k}$ is true, for a $k\in U$, then we must have exactly $n$ elements;
% that is, $\spec_{\T_{m}^{n}}(P_{k})=\{n\}$ if $k\in U$, and $\{m,n\}$ otherwise.
The theory $\T_{\leq}^{S}$, dependent on an infinite computable set $S \subseteq \No$, has only interpretations of cardinality in $S \cup \{\aleph_0\}$, and if $P_{n}$ is true and $F(n)$ is finite, then we have at most $F(n)$ elements;
so $\spec_{\T}(P_{n})=S\cap[1,F(n)]$ if $F(n)$ is finite, and $S\cup\{\aleph_{0}\}$ otherwise.
$\Tinfty$ is simple:
it has only infinite interpretations.
% The only new test theory for this paper is the $\Spnn$-theory $\Testcbfs$:
% for any predicate $P_{m,n}$ with $1\leq n\leq m$, an interpretation where $P_{m,n}$ holds has either exactly $n$ elements, or at least $m+1$ elements;
% much like $\T_{=}^{P}$ it helps us compute whether an element is in the spectrum of a formula, but it does that in a way such that the spectra of all its satisfiable quantifier-free formulas is cofinite.
One more test theory, $Th_{\T}$, is defined below as its axiomatization would not fit in the table.

\begin{definition}\label{def:theorytheory}
Let $\T$ be an arbitrary $\Sigma$-theory with computable finite spectra. Given a  quantifier-free $\Sigma$-formula $\phi$, let $S_\phi = \No \setminus \spec_\T(\phi)$. 
For each $i\in\No$, let 
$s_{\phi,i}$ be its $i$th element in increasing order,
or $\Inf$ if $S_\phi$ has fewer than $i$ elements. Now, let $Th_\T$ be the $\Sigmasigma$-theory axiomatized by
\begin{align*}
    &\{P_{\phi,n} \rightarrow \psi_{=s_{\phi,n}} : \phi \in \qf{\Sigma},\; n \in \No,\; s_{\phi,n} < \Inf\} \cup \\
    &\{P_{\phi,n} \rightarrow \psi_{>m} : \phi \in \qf{\Sigma},\; m,n \in \No,\; s_{\phi,n} = \Inf\} \cup \\
    &\{P_{\phi,n} \rightarrow \lnot P_{\phi',n'} : (\phi,n) \neq (\phi',n')\}.
\end{align*}
\end{definition}

So, in a $Th_{\T}$-interpretation, the truth of $P_{\phi,n}$ implies the interpretation has exactly $s_{\phi,n}$ elements (except when $s_{\phi,n}=\aleph_{0}$, when the interpretation can have any infinite number of elements).
Notice that in all theories but $\Tinfty$ (which has no $0$-ary predicates) and $\T_{=}^{P}$ (where this is implied by the axioms), we also demand that no two $0$-ary predicates can be simultaneously true:
this makes the proofs related to them much simpler and reasonably similar to each other.

\begin{restatable}{proposition}{tableiscorrect}\label{theo:test-theories}
    For each theory in
    \Cref{tab:testtheoriesproperties}, the properties that are listed under it hold.
\end{restatable}

% \begin{proposition}\label{prop:list-of-properties}
%     \begin{enumerate}
%         \item $\T^{P}_{>n}$ has computable bounded spectra, and thus computable bounded finite spectra.
%         \item $\Tinfty$ has computable bounded spectra, and is stably infinite with computable spectra.
%         \item $\T^{P}_{=}$ has computable bounded spectra, and thus computable bounded finite spectra.
%         \item $\T^{S}_{\leq}$ has computable bounded finite spectra.
%         \item $Th_{\T}$ has computable bounded finite spectra.
%     \end{enumerate}
% \end{proposition}

\begin{proof}
We start with the theories defined in \cite{POPL}.
    \begin{enumerate}
        \item That $\T^{P}_{>n}$ is decidable and stably infinite was proven in \cite{POPL}, and stable infiniteness implies having computable bounded spectra (which in turn implies having computable bounded finite spectra) by \Cref{lem:SI->CBS}.

        Now take a $\T$-satisfiable, quantifier-free $\phi$ without $\aleph_{0}$ in its spectrum, and by \Cref{compactness2} we get there exists a finite $n=\max\spec_{\T}(\phi)$.
        If $\phi\wedge P_{m}$ is $\T\sqcup\T^{P}_{>n}$-satisfiable, then $P_{m}$ must have a model with at most $n$ elements, and by the axiomatization of $\T^{P}_{>n}$ we get $m\notin U$.
        Reciprocally, if $m\notin U$, we take a $\T$-interpretation that satisfies $\phi$ with $n$ elements, make $P_{m}$ true and all other $P_{k}$ false, and the result is a $\T\sqcup\T^{P}_{>n}$-interpretation.

        \item \cite{POPL} proved that $\T^{P}_{=}$ is decidable and has computable spectra, and this last property implies having computable bounded spectra by \Cref{lem:CFS->CBFS}, and the facts that having computable spectra is equivalent to being infinitely decidable and having computable finite spectra, and that having computable bounded spectra is equivalent to being infinitely decidable and having computable bounded finite spectra.

        Suppose $\phi\wedge P_{n}$ is $\T\sqcup\T^{P}_{=}$-satisfiable:
        since $P_{n}$ only has models of cardinality $n$, from the axiom $P_{n}\rightarrow\psi_{=n}$, we get $n\in \spec_{\T}(\phi)$.
        Reciprocally, if $n\in\spec_{\T}(\phi)$, we take a $\T$-interpretation that satisfies $\phi$ with $n$ elements, make $P_{n}$ true in it and all other $P_{k}$ false, and the result is a $\T\sqcup\T^{P}_{=}$-interpretation that satisfies $\phi\wedge P_{n}$.

        \item $\T_{\leq}^{S}$ is decidable and has computable finite spectra as proven in \cite{POPL}, and \Cref{lem:CFS->CBFS} shows the latter implies having computable bounded finite spectra.

        Next we suppose $S=\No\setminus\spec_{\T}(\phi)$ is infinite and $\aleph_{0}\in\spec_{\T}(\phi)$.
        If $\phi\wedge P_{k}$ is $\T\sqcup\T_{\leq}^{S}$-satisfiable, since $\T_{\leq}^{S}$ only has models of cardinality in $S$ (from the axioms $\neg\psi_{=n}$ for $n\notin S$) or infinite, and $\phi$ only has models of cardinality not in $S$, we get $P_{k}$ must have an infinite model, meaning $F(k)=\aleph_{0}$.
        Reciprocally, if $F(k)=\aleph_{0}$, there is an infinite $\T_{\leq}^{S}$-interpretation that satisfies $P_{k}$, and since $\phi$ has an infinite model by hypothesis we get $\phi\wedge P_{k}$ is $\T\sqcup\T_{\leq}^{S}$-satisfiable.
        
        \item $\Tinfty$ is decidable and stably infinite, and has computable spectra by \cite{POPL}, and, again, both of these properties imply having computable bounded spectra by \Cref{lem:SI->CBS}.

        Now, if $\phi$ is $\T\sqcup\Tinfty$-satisfiable, since $\Tinfty$ only has infinite models we get $\aleph_{0}\in\spec_{\T}(\phi)$.
        Reciprocally, if $\aleph_{0}\in\spec_{\T}(\phi)$, there is an infinite $\T$-interpretation that satisfies $\phi$, and of course this interpretation doubles as a $\Tinfty$-interpretation, meaning $\phi$ is $\T\sqcup\Tinfty$-satisfiable.

        \item It was shown in \cite{POPL} that $Th_{\T}$ is decidable and has computable finite spectra, and, again, this property implies having computable bounded finite spectra as shown in \Cref{lem:CFS->CBFS}.

        Now suppose $\phi\wedge P_{\phi,n}$ is $\T\sqcup Th_{\T}$-satisfiable, so $\phi$ and $P_{\phi,n}$ have models of equal cardinalities:
        if $S_{\phi}=\No\setminus\spec_{\T}(\phi)$\footnote{Notice this is a set of finite numbers.} has $n$ elements or more, and $s_{\phi,i}$ is its $i$th element (if it has one, otherwise $s_{\phi,i}=\aleph_{0}$), we have $s_{\phi,n}$ must be finite;
        from the axiom $P_{\phi,n}\rightarrow\psi_{=s_{\phi,n}}$, valid as $s_{\phi,n}$ is finite, we obtain $P_{\phi,n}$ has only models of cardinality $s_{\phi,n}$.
        This leads to a contradiction, as $s_{\phi,n}$ is not in the spectrum of $\phi$, and we must therefore conclude $|\No\setminus\spec_{\T}(\phi)|<n$.

        Conversely, let $S_{\phi}=\No\setminus\spec_{\T}(\phi)$ and suppose $|S_{\phi}|<n$:
        if $s_{\phi,i}$ is the $i$th element of $S_{\phi}$ if it exists, and $\aleph_{0}$ otherwise, we get $s_{\phi,n}=\aleph_{0}$.
        The fact that $|S_{\phi}|<n$ furthermore implies $\spec_{\T}(\phi)$ is infinite, and by \Cref{compactness2} this means $\aleph_{0}\in\spec_{\T}(\phi)$, so take an infinite $\T$-interpretation $\A$ that satisfies $\phi$, make $P_{\phi,n}$ true and all other $P_{\phi,m}$ false.
        From the axioms $P_{\phi,n}\rightarrow \psi_{\geq m}$, valid when $s_{\phi,n}=\aleph_{0}$, we conclude $\A$ is also a $Th_{\T}$-interpretation now, so $\phi\wedge P_{\phi,n}$ is $\T\sqcup Th_{\T}$-satisfiable.
    \end{enumerate}

    The proofs for $\Testcbfs$ can be found in \Cref{prop:prop-testcbfs}.
\end{proof}

\section{Proof of \Cref{thm:G-cbfs}}

\IIpartCBFS*

\begin{proof}
    \begin{enumerate}
        \item Suppose that $\T$ is decidable and can be combined with every decidable theory with computable bounded finite spectra.
        It has computable finite spectra, as $P_{k}\wedge \phi$ is $\T\sqcup\Testcfs$-satisfiable if and only if $k\in \spec_{\T}(\phi)$, and $\Testcfs$ is decidable and has computable bounded finite spectra (as shown in \Cref{tab:testtheoriesproperties}), due to \Cref{lem:CFS->CBFS} and the fact $\Testcfs$ has computable finite spectra.

        Now, we prove that $\T$ is stably infinite.
        If $\phi$ is $\T$-satisfiable but $\aleph_{0}\notin \spec_{\T}(\phi)$, by \Cref{compactness2} there is a maximum $n$ in $\spec_{\T}(\phi)$.
        Thus $P_{k}\wedge\phi$ is $\T\sqcup\Testsi$-satisfiable if and only if $k\notin U$, a contradiction as $\Testsi$ is decidable and has computable bounded finite spectra, as we stated in \Cref{tab:testtheoriesproperties}.

Next, we show gentleness. For that, we first
         show $\spec_\T(\phi)$ is cofinite, for all $\T$-satisfiable quantifier-free $\phi$. If it is not, then $S = \No \setminus \spec_\T(\phi)$ is infinite. 
        Since $\T$ has computable finite spectra, $S$ is a decidable set. 
        Yet $P_{k}\wedge\phi$ is $\T\sqcup\Testg$-satisfiable if and only if $F(k)=\Inf$ (see the proof of this in \Cref{theo:test-theories}), an undecidable problem, despite $\Testg$ being decidable and having computable bounded finite spectra, as stated in \Cref{tab:testtheoriesproperties}.

    We now know $\spec_{\T}(\phi)$ must be cofinite for every quantifier-free $\phi$, but we must still prove gentleness. 
    Since $\T$ has computable finite spectra, it suffices to show that, given $\phi$, we can compute the cardinality of $S = \No \setminus \spec_{\T}(\phi)$ (and then we can output $S$ by testing whether $n\in S$, starting at $n=1$, until we have $|S|$ elements, and otherwise we increase $n$ to $n+1$). 
    Checking \Cref{tab:testtheoriesproperties} we see $Th_{\T}$ is decidable and has computable bounded finite spectra, as it has computable finite spectra, so $\T \sqcup Th_{\T}$ is decidable. 
    Now, $\phi \wedge P_{\phi,n}$ is $\T \sqcup Th_{\T}$-satisfiable if and only if $n > |S|$. Thus, by testing formulas of this form for satisfiability, we can compute the cardinality of $S$.

        \item Suppose that $\T$ is decidable and can be combined with every decidable, gentle and stably infinite theory, and we take a quantifier-free $\phi$, and $m,n\in\No$:
        our algorithm simply returns whether $\phi\wedge P_{m,n}$ is $\T\sqcup\Testcbfs$-satisfiable;
        in case $\phi\wedge\neq(x_{1},\ldots,x_{m})$ is not $\T$-satisfiable this happens if and only if $n\in \spec_{\T}(\phi)$, for $1\leq n\leq m$, so the algorithm has the required properties.
        We are then done as $\T\sqcup\Testcbfs$ is indeed decidable, given that $\Testcbfs$ is decidable, gentle and stably infinite, see \Cref{tab:testtheoriesproperties}.

    \end{enumerate}
\end{proof}

\section{Proof of \Cref{thm:G-cbs}}

\thmGcbs*

\begin{proof}
    \begin{enumerate}
    
        \item Suppose that $\T$ is decidable and can be combined with every decidable theory with computable bounded spectra, but is not stably infinite.
        Then there is a $\T$-satisfiable quantifier-free $\phi$ such that $\aleph_{0}\notin\spec_{\T}(\phi)$, and by \Cref{compactness2} there is a maximum $n$ in $\spec_{\T}(\phi)$.
        Therefore $\phi\wedge P_{k}$ is $\T\sqcup\Testsi$-satisfiable iff $k\notin U$, despite the fact that $\Testsi$ is decidable and has computable bounded spectra, as stated in \Cref{tab:testtheoriesproperties}, and therefore $\T\sqcup\Testsi$ should be decidable, by \Cref{thm:cbs}.
        This means $\T$ must be stably infinite.
        
        Now, $\T$ has computable spectra, as:
        $\T$ is infinitely decidable, given that it is stably infinite and \Cref{prop:known-relations-properties};
        and we can compute whether $m\in\spec_{\T}(\phi)$ by checking if $\phi\wedge P_{m}$ is $\T\sqcup\Testcfs$-satisfiable, as $\Testcfs$ is decidable and has computable bounded spectra, as proven in \Cref{theo:test-theories}.

        \item Suppose now $\T$ is decidable and can be combined with every decidable, stably infinite theory with computable spectra.
        It is infinitely decidable, as $\aleph_{0}\in\spec_{\T}(\phi)$ if and only if $\phi$ is $\T\sqcup\Tinfty$-satisfiable, and $\Tinfty$ is decidable, stably infinite and has computable spectra (what was shown in \Cref{tab:testtheoriesproperties}).
        Suppose that $\aleph_{0}\notin \spec_{\T}(\phi)$, and by \Cref{compactness2} there exists a maximum $m$ in $\spec_{\T}(\phi)$;
        then $\phi\wedge P_{m,n}$ is $\T\sqcup\Testcbfs$-satisfiable iff $n\in\spec_{\T}(\phi)$, and so $\spec_{\T}(\phi)$ is computable, as $\Testcbfs$ is decidable, stably infinite and gentle, and thus has computable spectra by \Cref{prop:known-relations-properties}, as stated in \Cref{tab:testtheoriesproperties}.
    \end{enumerate}
\end{proof}

\section{Proof of \Cref{prop:bigselfcomb}}
\bigselfcomb*

\begin{proof}
    \begin{enumerate}
        \item We will use the quantifier-free formulas $P_{1}^{\prime}\wedge P_{k}$, for $k>1$.
        They are $\Th{cbfs}\sqcup\Th{cbfs}$-satisfiable iff $F(k)=\aleph_{0}$, and deciding whether that happens cannot be achieved algorithmically.

        Indeed, start by assuming $P_{1}^{\prime}\wedge P_{k}$ is $\Th{cbfs}\sqcup\Th{cbfs}$-satisfiable:
        as $P_{1}^{\prime}$ has only infinite models (by the axioms $P_{1}\rightarrow\psi_{\geq k}$), this implies $P_{k}$ has an infinite model, and therefore that $F(k)=\aleph_{0}$ (since, otherwise, $P_{k}\rightarrow \psi_{\leq F(k)}$).
        Conversely, if $F(k)=\aleph_{0}$, we transform an infinite set into a $\Th{cbfs}$-interpretation by making $P_{1}^{\prime}$ and $P_{k}$ true, and all other $P_{i}$ and $P_{j}^{\prime}$ false;
        of course the resulting interpretation satisfies $P_{1}^{\prime}\wedge P_{k}$.

    \item Consider the quantifier-free formulas $P_{1}^{\prime}\wedge P_{k}$.
    They are $\Th{cbs}\sqcup\Th{cbs}$-satisfiable if and only if $k\in U$, an undecidable problem, proving the theory is undecidable.

    Indeed, suppose first that $P_{1}^{\prime}\wedge P_{k}$ is $\Th{cbs}\sqcup\Th{cbs}$-satisfiable, and that $k\notin U$:
    from the axioms $P_{1}^{\prime}\rightarrow \neg\psi_{=k}$ and $P_{k}\rightarrow\psi_{=k}$, this leads to a contradiction.
    Conversely, suppose $k\in U$:
    we take an interpretation with $k$ elements in its domain, make $P_{1}^{\prime}$ true and all other $P_{i}^{\prime}$ false, $P_{k}$ true and all other $P_{i}$ false, and the result is a $\Th{cbs}\sqcup\Th{cbs}$-interpretation that satisfies $P_{1}^{\prime}\wedge P_{k}$.

        \item Consider the quantifier-free formulas $P_{1}^{\prime}\wedge P_{k}$, for $k>1$:
        they are $\Th{cfs}\sqcup\Th{cfs}$-satisfiable iff $F(k)=\aleph_{0}$, which constitutes an undecidable problem.

        Suppose first that $P_{1}^{\prime}\wedge P_{k}$ is $\Th{cfs}\sqcup\Th{cfs}$-satisfiable:
        since $P_{1}^{\prime}$ has only infinite models, this means $P_{k}$ has an infinite model, which only happens if $F(k)=\aleph_{0}$.
        Conversely, if $F(k)=\aleph_{0}$, we take an infinite set, make $P_{1}^{\prime}$ and $P_{k}$ true and all other $P_{i}$ and $P_{j}^{\prime}$ false:
        the result is a $\Th{cfs}$-interpretation that satisfies $P_{1}^{\prime}\wedge P_{k}$.
        % \item (Non self-combinability of $\T(m,n)$ was proven in \Cref{Tmnnonself}.)
    \end{enumerate}
\end{proof}

\end{document}

%% file: mycommands.tex
\newcommand{\Exists}[1]{\exists\,#1.\:}

\newcommand{\Forall}[1]{\forall\,#1.\:}

\newcommand{\A}{\ensuremath{\mathcal{A}}}

\newcommand{\B}{\ensuremath{\mathcal{B}}}

\newcommand{\T}{\ensuremath{\mathcal{T}}}

\newcommand{\G}{\ensuremath{gen}}

\newcommand{\N}{\ensuremath{\mathbb{N}_{\omega}}}

\newcommand{\qf}[1]{\ensuremath{QF(#1)}}

\newcommand{\minmod}{\ensuremath{\mathrm{minmod}}}

\newcommand{\Tgeqn}{\ensuremath{\T_{\geq n}}}

\newcommand{\Tinfty}{\ensuremath{\T_{\infty}}}

\newcommand{\Tleqn}{\ensuremath{\T_{\leq n}}}

\newcommand{\NNEQ}[1]{\ensuremath{\neq(#1_{1},\ldots,#1_{n})}}

\newcommand{\distinct}[1]{\NNEQ{x}}

\renewcommand{\int}[2]{\mathcal{#1}/\mathcal{#2}}

\newcommand{\Sp}{\ensuremath{\Sigma_{P}}}

\newcommand{\Spn}{\ensuremath{\Sigma_{P}^{\mathbb{N}}}}

\newcommand{\Spnn}{\ensuremath{\Sigma_{P}^{\mathbb{N} \times \mathbb{N}}}}

\newcommand{\No}{\ensuremath{\mathbb{N}^{*}}}

\newcommand{\Inf}{\ensuremath{\aleph_{0}}}

\newcommand{\spec}{\ensuremath{\textit{Spec}}}

\newcommand{\dom}[1]{\ensuremath{dom(#1)}}

\newcommand{\Th}[1]{\ensuremath{\T_{\langle\scriptscriptstyle #1\rangle}}}

\newcommand{\class}[1]{\ensuremath{\mathfrak{T}_{\text{#1}}}}

\newcommand{\Testsi}{\ensuremath{\T_{>n}^{P}}}

\newcommand{\Testcfs}{\ensuremath{\T_{=}^{P}}}

\newcommand{\Testg}{\ensuremath{\T_{\leq}^{S}}}

\newcommand{\Testcbfs}{\ensuremath{\T_{=}^{>}}}

\newcommand{\Sigmasigma}{\ensuremath{Sig^{P}_{\Sigma}}}

\newcommand{\Gal}{\ensuremath{G}}

\newcommand{\unc}{\ensuremath{U}}

%% file: bib.bib
@article{OppenSI,
title = {Complexity, convexity and combinations of theories},
journal = {Theoretical Computer Science},
volume = {12},
number = {3},
pages = {291-302},
year = {1980},
issn = {0304-3975},
doi = {https://doi.org/10.1016/0304-3975(80)90059-6},
author = {Derek C. Oppen}
}

@InProceedings{CADE,
author="Toledo, Guilherme V.
and Zohar, Yoni
and Barrett, Clark",
editor="Pientka, Brigitte
and Tinelli, Cesare",
title="Combining Combination Properties: An Analysis of Stable Infiniteness, Convexity, and Politeness",
booktitle="Automated Deduction -- CADE 29",
year="2023",
publisher="Springer Nature Switzerland",
address="Cham",
pages="522--541"
}

@InProceedings{POPL,
author="Przybocki, Benjamin 
and Toledo, Guilherme V.
and Zohar, Yoni",
title="Characterizing Sets of Theories That Can Be Disjointly Combined",
booktitle=" 53rd ACM SIGPLAN Symposium on Principles of Programming Languages (POPL 2026)",
year="2026"
}

@inproceedings{CasalRasga,
  author    = {Filipe Casal and
               Jo{\~{a}}o Rasga},
  title     = {Revisiting the Equivalence of Shininess and Politeness},
  booktitle = {{LPAR}},
  series    = {Lecture Notes in Computer Science},
  volume    = {8312},
  pages     = {198--212},
  publisher = {Springer},
  year      = {2013},
  doi={10.1007/978-3-642-45221-5_15}
}

@book {marker2002,
    AUTHOR = {Marker, David},
     TITLE = {Model theory: {An} introduction},
    SERIES = {Graduate Texts in Mathematics},
    VOLUME = {217},
 PUBLISHER = {Springer-Verlag, New York},
      YEAR = {2002},
     PAGES = {viii+342},
      ISBN = {0-387-98760-6},
}

@InProceedings{gentle,
author="Fontaine, Pascal",
editor="Ghilardi, Silvio
and Sebastiani, Roberto",
title="Combinations of Theories for Decidable Fragments of First-Order Logic",
booktitle="Frontiers of Combining Systems",
year="2009",
publisher="Springer Berlin Heidelberg",
address="Berlin, Heidelberg",
pages="263--278",
isbn="978-3-642-04222-5"
}

@TECHREPORT{shiny,
  author =	 {Cesare Tinelli and Calogero Zarba},
  title =	 {Combining decision procedures for theories in sorted logics},
  institution =	 {Department of Computer Science, The University of Iowa},
  number = 	 {04-01},
  month =	 feb,
  year =	 2004,
  url =          {ftp://ftp.cs.uiowa.edu/pub/tinelli/papers/TinZar-RR-04.pdf},
}

@InProceedings{polite,
    TITLE = {Combining data structures with nonstably infinite theories using many-sorted logic},
    AUTHOR = {Ranise, Silvio and Ringeissen, Christophe and Zarba, Calogero G.},
    URL = {https://hal.inria.fr/inria-00000570},
    BOOKTITLE = {5th International Workshop on Frontiers of Combining Systems - FroCoS'05},
    ADDRESS = {Vienna},
    EDITOR = {Bernard Gramlich},
    PUBLISHER = {Springer},
    SERIES = {Lecture Notes in Artificial Intelligence},
    VOLUME = {3717},
    PAGES = {48--64},
    YEAR = {2005},
    MONTH = Sep,
    DOI = {10.1007/11559306}
}

@article{NelsonOppen,
author = {Nelson, Greg and Oppen, Derek C.},
title = {Simplification by Cooperating Decision Procedures},
year = {1979},
issue_date = {Oct. 1979},
publisher = {Association for Computing Machinery},
address = {New York, NY, USA},
volume = {1},
number = {2},
issn = {0164-0925},
url = {https://doi.org/10.1145/357073.357079},
doi = {10.1145/357073.357079},
journal = {ACM Trans. Program. Lang. Syst.},
month = oct,
pages = {245–257},
numpages = {13}
}

@incollection{BSST21,
   author = {Clark Barrett and Roberto Sebastiani and Sanjit Seshia and
	Cesare Tinelli},
   editor = {Armin Biere and Marijn J. H. Heule and Hans van Maaren and
	Toby Walsh},
   title = {Satisfiability Modulo Theories},
   booktitle = {Handbook of Satisfiability, Second Edition},
   series = {Frontiers in Artificial Intelligence and Applications},
   volume = {336},
   chapter = {33},
   pages = {825--885},
   publisher = {IOS Press},
   month = feb,
   year = {2021},
   url = {http://www.cs.stanford.edu/~barrett/pubs/BSST21.pdf}
}

@inproceedings{JB10-LPAR,
  author       = {Dejan Jovanovic and
                  Clark W. Barrett},
  title        = {Polite Theories Revisited},
  booktitle    = {{LPAR} (Yogyakarta)},
  series       = {Lecture Notes in Computer Science},
  volume       = {6397},
  pages        = {402--416},
  publisher    = {Springer},
  year         = {2010}
}

@article{bonacina2019theory,
  title={Theory combination: beyond equality sharing},
  author={Bonacina, Maria Paola and Fontaine, Pascal and Ringeissen, Christophe and Tinelli, Cesare},
  journal={Description Logic, Theory Combination, and All That: Essays Dedicated to Franz Baader on the Occasion of His 60th Birthday},
  pages={57--89},
  year={2019},
  publisher={Springer}
}

@inproceedings{10.1007/978-3-031-99984-0_2,
author = {V. Toledo, Guilherme and Przybocki, Benjamin and Zohar, Yoni},
title = {Being Polite Is Not Enough (and Other Limits of Theory Combination)},
year = {2025},
isbn = {978-3-031-99983-3},
publisher = {Springer-Verlag},
address = {Berlin, Heidelberg},
url = {https://doi.org/10.1007/978-3-031-99984-0_2},
doi = {10.1007/978-3-031-99984-0_2},
booktitle = {Automated Deduction – CADE 30: 30th International Conference on Automated Deduction, Stuttgart, Germany, July 28-31, 2025, Proceedings},
pages = {17–34},
numpages = {18},
location = {Stuttgart, Germany}
}

@book{Enderton,
	publisher = {New York, Academic Press},
	author = {Herbert Bruce Enderton},
	year = {1972},
	title = {A Mathematical Introduction to Logic},
isbn={978-0122384523}
}
